\documentclass[preprint,12pt]{elsarticle}
\usepackage[utf8]{inputenc}
\usepackage{amsthm}
\usepackage{amsmath}
\usepackage{amsfonts}
\usepackage{amssymb}
\usepackage{graphicx}
\usepackage{epstopdf}
\usepackage{color}
\usepackage{multirow}
\usepackage{mathtools}
\usepackage{xcolor}
\usepackage{a4wide}
\usepackage{bm}
\usepackage{caption}
\usepackage{subcaption}
\usepackage[
  separate-uncertainty = true,
  multi-part-units = repeat
]{siunitx}
\usepackage{rotating}
\usepackage{verbatim}
\usepackage{lineno,hyperref}
\modulolinenumbers[5]
\usepackage{tikz}
\tikzstyle{noeud} = [circle, draw, fill=white, inner sep=2pt]
\journal{Theoretical Computer Science}
\newtheorem{theorem}{Theorem}[section]
\newtheorem{lemma}{Lemma}[section]

\newdefinition{remark}{Remark}[section]

\newdefinition{definition}{Definition} \newdefinition{example}{Example}

\usepackage[noend,linesnumbered,ruled,lined]{algorithm2e}
\SetAlFnt{\scriptsize}
\usepackage{amsfonts}
\usepackage{makecell}

\usepackage[usenames,dvipsnames]{pstricks}
\usepackage{pstricks-add}
\usepackage{epsfig}
\usepackage{pst-grad} 
\usepackage{pst-plot} 
\usepackage[space]{grffile} 
\usepackage{etoolbox} 
\makeatletter 
\patchcmd\Gread@eps{\@inputcheck#1 }{\@inputcheck"#1"\relax}{}{}
\makeatother
\usepackage{mathtools}

\usepackage{graphicx}
\usepackage{tabularx}
\usepackage{textcomp}
\usepackage{xcolor}
\usepackage[labelformat=simple]{subcaption}

\usepackage{booktabs}

\usepackage{pgf,tikz,pgfplots}
\pgfplotsset{compat=1.15}
\usepackage{mathrsfs}
\usepackage{hyperref}
\usepackage{multicol}

\SetCommentSty{mycommfont}

\journal{ArXiv}

\usepackage{a4wide}
\begin{document}
\usetikzlibrary{arrows}
\begin{frontmatter}

\title{Cops and Lethal Robber on Rings: A Distributed Perspective}

\author[1]{Amanpreet Singh Saini}
\author[1]{Ashish Saxena }
\author[1]{Kaushik Mondal}
\affiliation[1]{organization={Department of Mathematics},
            addressline={Indian Institute of Technology Ropar}, 
            city={Rupnagar},
            postcode={140001}, 
            state={Punjab},
            country={India}}

\begin{abstract}
The Cops and Robber game is extensively studied in the sequential setting where the main goal is to capture the robber. Capturing the robber means at least one cop, and the robber will be at the same vertex together at some time.  There are several variants, including variants where the goal of the cops is to surround the robber. Surrounding the robber means there is at least one cop in each of the neighboring vertices of the robber's position.  In this paper, we introduce it in the distributed setting while empowering the robber by saying it can even kill cops. Specifically, in our model, the robber moves in odd rounds and has unbounded speed, cops move in even rounds, and if one or more cops move into a vertex where the robber is currently residing, all these cops get killed. We call this {\it lethal robber}. This also connects our work to the Intruder Capture and Black Hole Search problems by introducing an entity which is dynamic as well as lethal, a setting that, to the best of our knowledge, has not been studied.

In this work, we introduce the lethal robber, define the {\it cops and lethal robber} problem in the distributed setting and study it on a static ring of size $n$. We prove $n$ cops are not enough, even if all start from the same vertex, and provide an algorithm starting from an arbitrary initial configuration that requires $n+\lfloor\log n \rfloor+4$ cops in the worst case. 

\end{abstract}

\end{frontmatter}

\section{Introduction}
The \textit{Cops and Robber} is a pursuit-evasion game played in rounds on a finite graph $G$ between a set of $k\geq 1$ cops and a single robber. In its classical form, before starting the game, an initial position on the vertices of $G$ is chosen first by the cops, then by the robber. Then, in each round, first the cops, then the robber, move to neighbouring vertices or (if allowed by the variant of the game) stay in the current location. The game ends if the cops capture the robber. That is, the robber and at least one cop occupy the same vertex, in which case the cops have won. The robber wins by forever avoiding capture; note that, in this case, the game never ends. Recently, Jungeblut et al. \cite{Jungeblut_2025} introduced a variant that replaces capture with a stronger requirement: the cops must surround the robber, i.e., occupy all neighbours of the robber’s current position, thereby preventing any escape. This surrounding variant captures situations where containment, rather than direct capture, is the objective. 

Most existing work assumes that interaction with the robber is safe, meaning that a cop can freely move onto the robber’s vertex without consequences. In this work, we depart from this assumption and consider a setting where the robber behaves as a \textit{lethal entity}: any cop that moves onto the robber’s position is immediately destroyed without leaving any trace. At the same time, we restrict the robber’s power by disallowing it from moving onto vertices occupied by cops. Indeed, if such moves were permitted, the robber could eliminate all cops, making any meaningful notion of containment or surrounding impossible. This restriction ensures a well-defined and non-trivial objective, where the cops must coordinate to safely surround the robber. We call this problem the \textit{Cops and Lethal Robber game} ($\mathcal{CLR}$).

A related variant has been studied in~\cite{Bonato_2013}, where the robber is allowed to move onto a vertex occupied by cops and eliminate one of them; in particular, if two cops are present, the robber may eliminate one while the other can still capture it. In contrast, we adopt a stricter and more adversarial model: if multiple cops move onto the robber’s vertex simultaneously, all of them are destroyed. 

At the same time, $\mathcal{CLR}$ is closely related to the Intruder Capture (IC) problem \cite{Flocchini_2002}, where the objective is to capture an intruder that may move arbitrarily fast and is aware of the positions of all agents. The challenge is to design a strategy that guarantees the intruder's capture despite these advantages. Similarly, in $\mathcal{CLR}$, the robber is arbitrarily fast, and the cops must progressively restrict its mobility until it is eventually surrounded. As in IC, the robber's location is initially unknown. However, unlike IC, the robber in $\mathcal{CLR}$ is lethal, introducing a fundamentally new challenge, as any cop that encounters the robber is eliminated. A further connection can be drawn with the Black Hole Search (BHS) problem~\cite{Dobrev_2001}. In BHS, a black hole is a lethal vertex that destroys any mobile entity entering it, and the objective is to identify its location. The lethal robber in $\mathcal{CLR}$ can be viewed as a dynamic counterpart of a black hole: any cop entering the robber's vertex is destroyed. However, in contrast to BHS, where the lethal object is static, the robber is mobile, strategic, and actively attempts to avoid containment. Rather than only detecting the dangerous entity, a more meaningful objective is to \emph{capture} or \emph{block} its movement.
Thus, $\mathcal{CLR}$ combines aspects of pursuit-evasion and lethal-environment, introducing challenges that are absent in either setting alone.

In this work, we study $\mathcal{CLR}$ in the distributed setting. The details of the distributed model and problem definition are presented in the next section.

\subsection{Model and the problem definition}
\noindent \textbf{Network model:} The network is modeled as an undirected, anonymous, port-labeled simple graph, denoted by $G = (V, E)$, where $V$ and $E$ are the sets of vertices and edges, respectively. In this work, we consider $G$ as a ring $C_n$, where $n=|V|$. The vertices are anonymous, and at each vertex $v \in V$, the incident edges are assigned distinct port numbers from the set $\{0,1\}$. Port labels are local: an edge $(u,v) \in E$ may have different port numbers at its two endpoints. Such labeling is necessary, as without port numbers it is impossible to distinguish among outgoing edges at a vertex \cite{Dessmark2006}.

\noindent \textbf{Cop model:} We consider $k \geq 1$ cops placed arbitrarily on the vertices of the ring $C_n$ by an adversary. Each cop is assigned a unique identifier from the range $[1,n^\lambda]$, where $\lambda$ is a positive constant. The identifier of a cop $c$ is denoted by $c.\mathrm{ID}$. Each cop knows only its own identifier and has no information about the identifiers of other cops. Furthermore, the cops have no prior knowledge of the size or structure of the ring, the initial positions of other cops, or the distances between them. Each cop can distinguish between port 0 and port 1 at the node it resides at any round. Also, each cop is equipped with some memory.

The system operates in synchronous rounds, starting from round $0$. Following the classical setting, we assume that the cops act in even-numbered rounds. In each such round, a cop at a vertex $v$ can observe the degree of $v$ and the port numbers of all incident edges. During each even round, every cop executes a \textit{Communicate--Compute--Move} cycle, in which it first communicates with all other cops present at the same vertex, then performs local computation to decide whether to move and, if so, selects a port number, and finally moves through the selected port, if any. Upon entering a vertex, each cop $c$ records the port through which it arrived, denoted $c.portin$. Two cops crossing the same edge simultaneously in opposite directions cannot detect or communicate with each other. We say a configuration is \underline{rooted} if all cops are initially co-located at one vertex, and \underline{scattered} otherwise.

\noindent \textbf{Lethal robber model:} We consider a \emph{lethal robber}. Initially, the robber has complete knowledge of $C_n$, including the positions of all cops and their algorithms. This knowledge is available to the robber in every odd-numbered round. At a round, if the robber is at vertex $u$, it may either stay at $u$ or move to a vertex $v$ provided there exists a path from $u$ to $v$ with no vertex on that path occupied by any cop. The robber is also not allowed to move onto a vertex occupied by a cop. The robber has unbounded speed and thus can traverse any such path within a single round, but must move along the edges of the graph. The cops, on the other hand, do not know the location of the robber and can move only to an adjacent vertex in one round. If one or more cops move onto the vertex currently occupied by the robber, all such cops are destroyed without leaving any trace.  

\noindent \textbf{Problem definition:}
The $\mathcal{CLR}$ game is played in rounds on the ring $C_n$ between a set of $k \geq 1$ cops and a single lethal robber. Before the start of the game, the cops are placed on the vertices of $C_n$ by an adversary, after which the robber chooses its initial position on some vacant vertex. The cops have no knowledge of the robber’s location or the locations of the other cops. The game proceeds in synchronous rounds. The cops act in each even round, followed by the robber that acts in each odd round, according to the rules of the model. The cops win if they \emph{surround} the robber; that is, at some round $r$, the robber is located at a vertex $v \in V$ and all neighbors of $v$ are occupied by cops. In this case, the robber cannot move and is captured. The robber wins by either eliminating all cops or forever avoiding being surrounded; that is, if the robber is at a vertex $v$, then at least one neighbor of $v$ is not occupied by any cop. In the latter case, the game never ends.

The objective is to design a distributed strategy for the cops that guarantees winning the game from any arbitrary initial configuration.

\subsection{Related work}\label{sec:related work}
The study of pursuit--evasion games on graphs is centered around the classical \emph{Cops and Robber} model, introduced independently by Nowakowski and Winkler and by Quilliot~\cite{nowakowski1983vertex,quilliot1978jeux}. The extension to multiple cops and the notion of the cop number \(c(G)\) were formalized by Aigner and Fromme~\cite{aigner1984game}, who established foundational results such as the fact that three cops suffice for planar graphs. Since then, the classical model has been extensively studied; we refer to~\cite{bonato2011game,nowakowski2019game} for comprehensive overviews. A central open problem in this area is Meyniel’s conjecture, which asserts that \(c(G)=O(\sqrt{n})\) for any connected graph on \(n\) vertices~\cite{meyniel1985cop,bonato2014survey}. From a computational perspective, deciding whether \(c(G)\leq k\) is NP complete~\cite{goldstein1995complexity}.

Several variants of the classical Cops and Robber game have been proposed to study the effect of restricting the robber's movement. In the restrictive vertex model of Burgess et al.~\cite{burgess2019restricted}, the robber is forbidden from entering vertices occupied by cops, a framework later extended by Bradshaw et al.~\cite{bradshaw2020bounded}. This idea was generalized in the containment variant of Crytser et al.~\cite{crytser2020containment}, where cops occupy edges and block the robber's traversal. Related restrictive models have also been studied in other settings, such as the face-based variant for planar graphs introduced by Jungeblut et al.~\cite{ha2023face}; see also~\cite{schneider2023thesis}. While these works focus on limiting the robber's movement, comparatively less attention has been given to surrounding-based objectives, where the cops win by occupying all neighbors (or incident edges) of the robber. This notion was recently formalized by Jungeblut et al.~\cite{Jungeblut_2025}.

A closely related line of research is the Intruder Capture (IC) problem \cite{Flocchini_2002}, where the objective is not only to detect but also to capture the intruder in a network. The IC problem has been studied extensively in various network settings (e.g., \cite{Blin_2008,DERENIOWSKI_2011,Flocchini_2008,fomin2005,NISSE_2009}), and a comprehensive survey is available in \cite{Nisse2019}. Another related problem is Black Hole Search (BHS) \cite{Dobrev_2001}, which can be viewed as a static counterpart of infection detection: infected nodes do not propagate the infection, and the goal is to locate such a node while minimizing losses. 

The problem $\mathcal{CLR}$ lies at the confluence of the IC, BHS, and Cops and Robbers paradigms. While IC and BHS have been extensively investigated in distributed settings, the Cops and Robbers game has largely been studied from a centralized graph-theoretic perspective. To the best of our knowledge, distributed variants of pursuit–evasion games of this type have not been explored. Consequently, $\mathcal{CLR}$ opens a new avenue for studying pursuit–evasion strategies under distributed computational constraints, such as limited visibility, restricted communication, and incomplete knowledge of the network.

\subsection{Preliminaries and notations}
 For a graph $G$, the cop number $c(G)$ is the smallest integer $k$ such that $k$ cops have a winning strategy against a single robber on $G$.

We now define the notions of \emph{boundary vertices}, \emph{safe zone}, and \emph{unsafe zone}. At any round $t$, the \emph{boundary vertices} are the vertices closest to the robber that contain at least one cop. In general, there are exactly two boundary vertices, denoted by $b^{t}_1$ and $b^{t}_2$. However, if all cops are located at a single vertex, then there is exactly one boundary vertex. The \emph{safe zone} at round $t$ is the segment of the ring between the boundary vertices, including the boundary vertices, that does not contain the robber. Let $Z_t$ denote the set of vertices in the safe zone. Observe that the robber cannot occupy any vertex of the safe zone. Moreover, if $|Z_t|=n-1$, then the robber is surrounded by the cops. We denote the length of the safe zone as $\bar{\rho}$, where $\bar{\rho}=|Z_t|-1$. The \emph{unsafe zone} at round $t$ is the segment between the boundary vertices that contains the vertex occupied by the robber. The robber may occupy any vertex of the unsafe zone except the boundary vertices. We denote by $\rho$ the length of the unsafe zone in the initial configuration and define it as $\rho=n-\bar{\rho}$. Figure~\ref{fig:safe_zone} illustrates these notions.
\begin{figure}[h]
    \centering
    \includegraphics[width=0.4\linewidth]{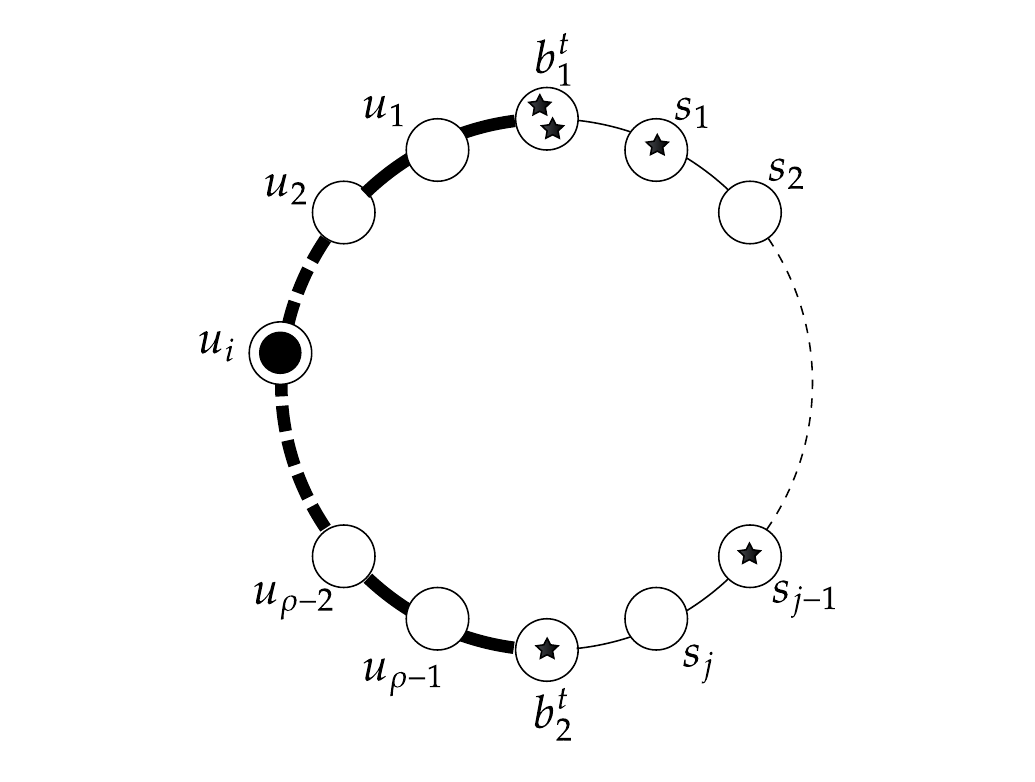}
   \caption{A configuration at round $t$ with boundary vertices $b_1^t$ and $b_2^t$. The segment $S_t=(b_1^t,s_1,s_2,\dots,s_j,b_2^t)$ denotes the safe zone (thin line), while $U_t=(u_1,u_2,\dots,u_{\rho-1})$ denotes the unsafe zone (thick line). The filled black circle at $u_i$ represents the robber, and a star inside a vertex represents a cop occupying that vertex.}
    \label{fig:safe_zone}
\end{figure}

\subsection{Our contribution}
In this work, we introduce the concept of lethal robber and study the Cops and Lethal Robber $\mathcal{CLR}$ game on the $n$-vertex ring $C_n$ in a distributed setting. 

We first show that $c(C_n)\geq n+1$ (refer  Theorem \ref{th:imp}). We then present a deterministic distributed algorithm that uses $\max\{\rho+\lfloor \log \rho \rfloor+4,\; n-\rho+4\}$ many cops starting from an arbitrary initial configuration, and successfully surrounds the robber (refer Theorem \ref{th:main}). Consequently,
$c(C_n)\le \max\{\rho+\lfloor \log \rho \rfloor+4,\; n-\rho+4\}$. Further, each cop requires $O(\log n)$ bits of memory, and the algorithm terminates within $O(n^3)$ rounds.

Our study also connects the IC and BHS problems. In IC, the intruder is dynamic, arbitrarily fast, and its location is unknown, but it is not lethal. In contrast, BHS considers a lethal but static node. Our model combines these two aspects by introducing a dynamic, arbitrarily fast, and lethal adversary, and can thus be viewed as a dynamic variant of BHS, a setting that, to the best of our knowledge, remains unexplored (see Section~\ref{sec:disc}).

\section{Impossibility with $n$ cops}\label{sec:imp}

In this section, we provide an impossibility result for the $\mathcal{CLR}$ problem, where $n \geq 4$. The intuition behind the result is that the robber, due to its unbounded speed and complete knowledge of the cops' strategy, can eliminate cops whenever they attempt to expand the safe zone. Consequently, exploring new vertices necessarily incurs losses.

Consider an initial configuration in which all $n$ cops are co-located at a single vertex, which is a special case of an arbitrary initial configuration. We analyze how the safe zone evolves as the cops attempt to explore new vertices. First, we bound the rate at which the safe zone can expand. We then show that every such expansion allows the robber to eliminate cops. Finally, combining these observations, we derive a lower bound on the number of cops required to successfully surround the robber.

\begin{lemma}\label{lem:atmax_1_inc}
For any even round $t$, the robber can always restrict the expansion of the safe zone to at most one new vertex in round $t+2$. Consequently, $|Z_{t+2}| \leq |Z_t|+1$.
\end{lemma}
\begin{proof}
At any even round $t$, the safe zone can expand only through the boundary vertices $b_1^t$ and $b_2^t$. Since each cop can move at most one hop per round, the cops can attempt to expand the safe zone by moving to the vertices adjacent to $b_1^t$ and $b_2^t$ in the unsafe zone during round $t+2$. Thus, they can attempt to add at most two new vertices to the safe zone.

If the cops attempt to expand the safe zone from only one boundary vertex, then the statement is immediate. Therefore, suppose they attempt to expand from both boundary vertices.
Since the robber has complete knowledge of the deterministic strategy of the cops, at round $t+1$ it knows which adjacent unsafe-zone vertices the cops will attempt to occupy from each boundary vertex in round $t+2$. The robber then moves to one of these two target vertices. Since the cops do not know the robber's location, so they execute the same moves, and the cop attempting to enter the occupied vertex is eliminated. Consequently, one of the two expansion attempts fails.
Therefore, regardless of the strategy of the cops, the robber can always restrict the expansion of the safe zone to at most one new vertex in round $t+2$. Hence, $|Z_{t+2}| \leq |Z_t|+1.$
\end{proof}

\begin{lemma}\label{lem:lossof1cop}
Let the initial configuration be rooted. For any $i\in[1,n-3]$ and any even round $t$, if $|Z_t|=i+1$, then by round $t$, the robber can always eliminate at least $i$ cops.
\end{lemma}
\begin{proof}
We prove the statement by induction on $i$.

\noindent\textbf{Base Case:}
Suppose that for some even round $t$, we have $|Z_t|=2$. By Lemma~\ref{lem:atmax_1_inc}, there exists an even round $t'\le t$ at which the safe zone expands from one vertex to two vertices.
If the cops attempt to expand the safe zone from only one boundary vertex, then the robber moves to the corresponding target vertex and eliminates the cop attempting to enter that vertex. Consequently, the safe zone does not expand. Therefore, to increase the safe zone from one vertex to two vertices, the cops must attempt to expand from both boundary vertices.

Since the robber knows the deterministic strategy of the cops, it can move to one of the target vertices and eliminate the cop attempting to enter that vertex in round $t'$. Hence, at least one cop is eliminated while the safe zone expands to two vertices. Thus, the statement holds for $i=1$.

\noindent\textbf{Inductive Step:}
Let $k\in[1,n-4]$, and assume that whenever $|Z_r|=k+1$ for some even round $r$, the robber can eliminate at least $k$ cops by round $r$.

Now consider an even round $t$ such that $|Z_t|=k+2$. By Lemma~\ref{lem:atmax_1_inc}, there exists an even round $t'\le t$ at which the safe zone expands from $k+1$ vertices to $k+2$ vertices.
As in the base case, the cops must attempt to expand from both boundary vertices; otherwise, the robber eliminates the cop attempting to enter the target vertex and the safe zone does not expand. Since the robber knows the deterministic strategy of the cops, it moves to one of the target vertices and eliminates the cop attempting to enter that vertex in round $t'$. Thus, one additional cop is eliminated while the safe zone expands from $k+1$ to $k+2$ vertices.

By the induction hypothesis, the robber has already eliminated at least $k$ cops before round $t'$. Therefore, by round $t$, it has eliminated at least $k+1$ cops.
Hence, by induction, whenever $|Z_t|=i+1$, the robber can eliminate at least $i$ cops by round $t$.
\end{proof}

\begin{lemma}\label{lem:atleast4alive}
Let $t$ be an even round such that $|Z_t| = n-2$ and $|Z_{t+2}| = n-1$. Then, at least four cops must be alive at round $t$.
\end{lemma}
\begin{proof}
Assume, for contradiction, that at most three cops are alive at round $t$. Since $|Z_t|=n-2$, the two boundary vertices $b^{t}_1$ and $b^{t}_2$ must each contain at least one cop. Hence, if at most three cops are alive at round $t$, then besides the cops occupying the boundary vertices $b^{t}_1$ and $b^{t}_2$, there can be at most one additional cop in the safe zone.
If this third cop is not located at a boundary vertex or at a vertex adjacent to a boundary vertex inside the safe zone, then it cannot contribute to the expansion of the safe zone in round $t+2$, since each cop can move at most one hop during a cops' turn. Therefore, without loss of generality, assume that the third cop is positioned either at $b^{t}_1$ or at a vertex adjacent to $b^{t}_1$ inside the safe zone.

Now, in order to expand the safe zone from $n-2$ vertices to $n-1$ vertices in round $t+2$, the cops must attempt to occupy a new vertex outside $Z_t$ from at least one of the boundary vertices. If the expansion is attempted from $b^{t}_2$ (or from both boundary vertices), then at round $t+1$ the robber positions itself on the vertex that a cop from $b^{t}_2$ will attempt to occupy in round $t+2$. Since the robber knows the deterministic strategy of the cops, it can predict this movement in advance and eliminate the moving cop in round $t+2$. On the other hand, if the expansion is attempted only from $b^{t}_1$, then the robber similarly positions itself to eliminate the cop moving from $b^{t}_1$.
Thus, in every possible case, the robber can prevent the cops from expanding the safe zone to $n-1$ vertices in round $t+2$. This contradicts our assumption that at most three cops are alive at round $t$ and yet the safe zone expands to $n-1$ vertices in round $t+2$. Therefore, at least four cops must be alive at round $t$.
\end{proof}

\begin{theorem}\label{th:imp}
For the $\mathcal{CLR}$ problem on $C_n$ with $n \geq 4$, we have $c(C_n)\geq n+1$.
\end{theorem}
\begin{proof}
Consider the rooted initial configuration in which all cops are initially located at a single vertex $v$. Then $|Z_0|=1$. To obtain a winning configuration, the cops must eventually reach an even round $t$ such that $|Z_t|=n-1$.

By Lemma~\ref{lem:atmax_1_inc}, the cops can expand the safe zone by at most one vertex in each cop's turn. Hence, before reaching a configuration with $|Z_t|=n-1$, the cops must first reach a configuration with $|Z_t|=n-2$. Now, Lemma~\ref{lem:lossof1cop}, which is proved for the rooted initial configuration, implies that by the time the safe zone expands to $n-2$ vertices, at least $n-3$ cops can be destroyed. Further, by Lemma~\ref{lem:atleast4alive}, at least four cops must be alive in order to expand the safe zone from $n-2$ vertices to $n-1$ vertices. Therefore, the total number of cops required is at least $(n-3)+4=n+1$. Hence, $c(C_n)\geq n+1.$ 
\end{proof}

\section{Cops and lethal robber}
In this section, we present \texttt{Capture-LR}, a deterministic distributed algorithm that surrounds the robber on $C_n$ whenever at least $\max\{\rho + \lfloor \log \rho \rfloor + 4,\; n - \rho + 4\}$ cops initiate execution from an arbitrary initial configuration, where $\rho$ denotes the length of the unsafe zone in the initial configuration. The cops have no prior knowledge of the robber's location or each other's positions. Throughout the paper, whenever we say that a cop performs an action in the next or a subsequent round, we mean the next even-numbered round in which the cops are permitted to move. We now present the high-level idea of the algorithm.

\medskip
\noindent\textbf{High-level idea.}
Initially the cops are arbitrarily distributed on the ring and have no knowledge of the robber's location or the boundary of the safe zone. The proposed algorithm proceeds in two steps. In Step~1, the objective is to identify the two boundary vertices of the safe zone and gather all surviving cops at these vertices. To achieve this, the cops collaboratively probe unexplored directions and from the outcomes of these probes, determine whether a direction is safe or leads toward the robber. Whenever both directions from an occupied vertex are verified to be safe, the available cops are redistributed to continue the exploration from other occupied vertices. This process continues until both boundary vertices are identified. At the completion of Step~1, one cop remains at each boundary vertex, while all remaining surviving cops are distributed between these two vertices as evenly as possible.

At the beginning of Step~2, every cop computes the same upper bound $\hat{\rho}$ of the unsafe zone, i.e., $\hat{\rho}\geq\rho$. Step~2 consists of at most $h=\lfloor\log\hat{\rho}\rfloor$ phases. Each phase, except the last, consists of three sub-phases: advancing, checking, and balancing, while the last phase contains only the advancing sub-phase. During the advancing sub-phase, a carefully chosen number of cops simultaneously advance from both boundary vertices into the unsafe zone. The checking sub-phase determines the actual extension achieved on each side and updates the boundary vertices accordingly. Finally, the balancing sub-phase redistributes the remaining cops between the new boundary vertices so that the next phase begins under the same conditions. Repeating this process progressively increases the safe zone until it spans the entire ring, and thereby surrounds the robber.

\subsection{Algorithm: \texttt{Capture-LR}}

In this section we present our algorithm \texttt{Capture-LR} along with the pseudo-codes.

\noindent\textbf{Step~1 (Determining the boundary).}
In Step~1, each cop assumes one of four roles: \texttt{settler}, \texttt{guard}, \texttt{helper}, or \texttt{checker}. The objective is to identify the two boundary vertices of the safe zone. At the completion of this step, exactly two cops serve as \texttt{guards}, one at each boundary vertex, and all remaining alive cops are \texttt{helpers} located at these vertices.

Initially, at every occupied vertex, the cop with the minimum identifier assumes the role of \texttt{settler}, while all remaining cops at that vertex, if any, become \texttt{helpers}.
Throughout Step~1, every cop $c$ maintains the parameters $c.role\in\{\texttt{settler},\texttt{guard},\texttt{helper},\texttt{checker}\}$ and $c.step\in\{1,2\}$, representing its current role and the current step of the algorithm, respectively. In addition, $c.\bar{\rho}$ stores the length of the safe zone in the initial configuration, and $c.N$ stores the total number of alive cops at the completion of Step~1. Initially, $c.step=1$, $c.\bar{\rho}=\perp$, and $c.N=\perp$.

Each \texttt{settler} and \texttt{guard} maintain the following parameters for each port $p \in \{0,1\}$. The parameter $c.safe_p \in \{\perp,0,1\}$ represents the safety status of port $p$, where $\perp$ denotes unknown, $1$ denotes safe, and $0$ denotes unsafe. The parameter $c.checkID_p$ stores the identifier of the \texttt{checker} assigned to port $p$, and equals $\perp$ when none is assigned. The parameter $c.dist_p \in \mathbb{N}$ denotes the current probing distance for port $p$. The parameter $c.time_p \in \mathbb{N}$ counts the rounds elapsed since the \texttt{checker} assigned to port $p$ last departed from $v$. Initially, $c.safe_p = \perp$, $c.checkID_p = \perp$, $c.dist_p = 0$, and $c.time_p = 0$ for both ports.

Each \texttt{checker} $c$ maintains two additional parameters: $c.dist_p \in \mathbb{N}$ and $c.operatorID$. The parameter $c.dist_p$ stores the distance currently being probed through port $p$, initialized to $1$ at the start of each checking task. The parameter $c.operatorID$ stores the identifier of the \texttt{settler} or \texttt{guard} that assigned $c$ to its current checking task, set at assignment and unchanged throughout.

Each \texttt{guard} additionally maintains three parameters: $c.lead \in \{\perp,0,1\}$, $c.uport \in \{0,1,\perp\}$, and $c.balance \in \{\perp,1\}$, all initialized to $\perp$. Upon learning the other \texttt{guard}'s identifier, $c$ sets $c.lead \leftarrow 1$ if its own identifier is smaller, and $c.lead \leftarrow 0$ otherwise. The parameter $c.uport$ stores the port leading toward the unsafe zone from $c$'s current vertex. The parameter $c.balance$ is set to $1$ by the \texttt{guard} with $c.lead = 1$ once it has confirmed that all surviving cops, except the other \texttt{guard}, have arrived at its vertex. We now define the procedures used in the pseudo-code of Step~1.
\begin{itemize}
    \item $\textsc{Assign}(p)$: A \texttt{settler} or \texttt{guard} $c$ at vertex $v$ executing this procedure selects the minimum-identifier \texttt{helper} at $v$, say $c_h$, and sets $c.checkID_p\leftarrow c_h.\mathrm{ID}$, $c.dist_p\leftarrow1$, and $c.time_p\leftarrow0$.

    \item $\textsc{MoveTill}(p,\mathit{role})$: A cop $c$ at vertex $v$ executing this procedure exits through port $p$ and continues in the same direction until it reaches a vertex occupied by a cop $c'$ with $c'.role=\mathit{role}$.

    \item $\textsc{Divide}()$: Let $\alpha$ denote the number of cops at the current vertex. A cop $c$ executing this procedure exits through port $0$ if it is among the $\lceil\alpha/2\rceil$ minimum-identifier cops at the current vertex, and through port $1$ otherwise. In either case, it continues in the chosen direction until it reaches a \texttt{settler} or a \texttt{guard}.

    \item $\textsc{Wait}(t)$: A cop $c$ executing this procedure remains idle for $t$ consecutive rounds, including the current round. Equivalently, if $c$ begins executing $\textsc{Wait}(t)$ in round $r$, it resumes execution of the subsequent statement at the beginning of round $r+t$.

    \item $\textsc{Check}(p,i)$: A cop $c$ at vertex $v$ executing this procedure traverses through port $p$, advancing one hop per round until reaching distance $i$, and then returns to $v$ along the same path. If another cop is encountered during the forward traversal, $c$ immediately returns to $v$. The parameter $c.meet_p$, initially $0$, is set to $1$ if such an encounter occurs.

    \item $\textsc{Extended\_Check}(p,i)$: A cop $c$ at vertex $v$ executing this procedure behaves as in $\textsc{Check}(p,i)$, except that it records the observed interaction using the parameter $c.enctype_p$, initially $0$. The value $0$ indicates that no cop, or only a \texttt{helper} or \texttt{checker}, is encountered. The value $1$ indicates that a \texttt{settler} currently having, or previously having had, at least one assigned \texttt{checker} is encountered. The value $2$ indicates that a \texttt{settler} that has never had an assigned \texttt{checker} is encountered. The value $3$ indicates that a \texttt{guard} is encountered, in which case its identifier is stored in $c.encID$. Throughout the traversal, $c$ sets $c.task\leftarrow\texttt{extcheck}$ so that any encountered cop can identify that it is executing this procedure.
\end{itemize}

\smallskip
\noindent\textbf{Algorithm for \texttt{checker}} (see Algorithm~\ref{alg:step1_checker}).
A \texttt{checker} is a \texttt{helper} temporarily assigned by the \texttt{settler} or the \texttt{guard} to probe a specific direction up to a specific distance.
If assigned by the \texttt{settler}, the \texttt{checker} first checks whether the port it is probing has already been marked safe; if so, it reverts to being a \texttt{helper}. Otherwise, it checks whether it encountered another cop during its last probe. If it did, the direction is safe, and the \texttt{checker} reverts to being a \texttt{helper}. If it did not, the \texttt{checker} probes one hop further. If assigned by the \texttt{guard}, the \texttt{checker} probes one hop further each time it returns by incrementing $c.dist_p$ and executing $\textsc{Extended\_Check}$, unless it met the other \texttt{guard} during its last probe; in that case, it reverts to being a \texttt{helper}.

\begin{algorithm}[!htb]
\DontPrintSemicolon
\LinesNumbered
\caption{\texttt{Capture-LR} : Step~1 (\texttt{checker})}
\label{alg:step1_checker}
Let $c$ be any cop at vertex $v$ with $c.step=1$ and $c.role=\texttt{checker}$.\;
\If{a cop $c_s$ with $c_s.checkID_p=c.\mathrm{ID}$, for some port $p$, is present at $v$}{
    \uIf{$c_s.role=\texttt{settler}$}{
        \uIf{$c_s.safe_p=1$}{
            $c.role\leftarrow\texttt{helper}$\;
        }
        \Else{
            \uIf{$c.meet_p=1$}{
                $c.role\leftarrow\texttt{helper}$\;
            }
            \Else{
                $c.dist_p\leftarrow c.dist_p+1$ and execute $\textsc{Check}(p,c.dist_p)$\;
            }
        }
    }
    \Else(\tcp*[f]{$c_s$ is a \texttt{guard}}){
        \uIf{$c.enctype_p\neq3$}{
            $c.dist_p\leftarrow c.dist_p+1$ and execute $\textsc{Extended\_Check}(p,c.dist_p)$\;
        }
        \Else{
            $c.role\leftarrow\texttt{helper}$\;
        }
    }
}
\end{algorithm}
\begin{algorithm}[!htb]
\DontPrintSemicolon
\LinesNumbered
\caption{\texttt{Capture-LR} : Step~1 (\texttt{helper}) }
\label{alg:step1_helper}
Let $c$ be any cop at vertex $v$ with $c.step=1$ and $c.role=\texttt{helper}$. Let $c_s$ be the \texttt{settler} or \texttt{guard} at $v$\;

\uIf{$c_s.safe_0=0$ and $c_s.safe_1=0$}{
    $c.N\leftarrow$ total cops at $v$, \: $c.\bar\rho\leftarrow 0$\;
    \If{$c$ is the minimum-ID \texttt{helper} at $v$}{
        $c.lead\leftarrow 0$, \: $c.uport\leftarrow 1$, \: $c.role\leftarrow\texttt{guard}$\;
    }
    $c.step\leftarrow 2$\;
}
\ElseIf{$c_s.safe_0=1$ and $c_s.safe_1=1$}{
    \uIf{$c_s.checkID_p\ne \perp$, for any port $p \in \{0,1\}$}{
        remain idle until $c_s.checkID_p = \perp$\;
    }
    execute $\textsc{Divide}()$\;
}
\Else{
    \uIf{$c_s.role=\texttt{settler}$}{
        \uIf{$c_s.safe_p=\perp$ and $c_s.checkID_p= \perp$ for both $p=0,1$}{
            \uIf{$c$ is the minimum-ID \texttt{helper} present at $v$}{
                $c.role\leftarrow\texttt{checker}$ and execute $\textsc{Check}(0,1)$\;}
            \ElseIf{$c$ is the second minimum-ID \texttt{helper} present at $v$}{
                $c.role\leftarrow\texttt{checker}$ and execute $\textsc{Check}(1,1)$\;}
        }
        \ElseIf{$c_s.safe_p=\perp$ and $c_s.checkID_p= \perp$ for exactly one  port $p$}{
            \uIf{$c$ is the minimum-ID \texttt{helper} present at $v$}{
                $c.role\leftarrow\texttt{checker}$ and execute $\textsc{Check}(p,1)$\;
            }
        }
      
    }
    \Else{
        \uIf{$c_s.balance=1$}{
            $c.N\leftarrow(\text{cops at }v')+1$,\: $c.\bar\rho\leftarrow c_s.\bar\rho$\; 
            \uIf{$c$ is among the $\lceil(c.N-2)/2\rceil$ minimum-ID helpers}{
                execute $\textsc{Wait}(c.\bar \rho)$\;
            }
            \Else{
                execute $\textsc{MoveTill}(1-c_s.uport,\texttt{guard})$\;
            }
            $c.step\leftarrow 2$\;
        }
        \Else(\tcp*[f]{not all cops have gathered yet}){
            \uIf{$c_s.safe_p=\perp$ and $c_s.checkID_p= \perp$ for some port $p$}{
                \uIf{$c$ is the minimum-ID \texttt{helper} present at $v$}{
                    $c.role\leftarrow\texttt{checker}$ and execute $\textsc{Extended\_Check}(p,1)$\;
                }
            }
            \uIf{a \texttt{checker} $c_{ckr}$ is present at $v$ with $c_{ckr}.portin=p$ and $c_{ckr}.task=\texttt{extcheck}$}{
                \uIf{$c_{ckr}.operatorID \ne c_s.\mathrm{ID}$}{
                    \uIf{$c_{ckr}.operatorID < c_s.\mathrm{ID}$}{
                        execute $\textsc{MoveTill}(p,\texttt{guard})$\;
                    }    
                }
                \Else{
                    \uIf{$c_{ckr}.enctype_p=2$}{
                        \uIf{$c$ is the minimum-ID helper}{
                            execute $\textsc{MoveTill}(p,\texttt{settler})$\;}
    
                    }
                    \uElseIf{$c_{ckr}.enctype_p=3$}{
                        \uIf{$c_s.\mathrm{ID}>c_{ckr}.encID$}{
                            execute $\textsc{MoveTill}(p,\texttt{guard})$\;
                        }
                    }
                }
            }
            \uIf{$c_s.lead=0$}{
                execute $\textsc{MoveTill}(1-c_s.uport,\texttt{guard})$\;
            }
        } 
    }
}
\end{algorithm}

\smallskip
\noindent\textbf{Algorithm for \texttt{helper}} (see Algorithm~\ref{alg:step1_helper}).
A \texttt{helper} assists the \texttt{settler} or \texttt{guard} at its current vertex $v$. If both ports of $v$ are confirmed unsafe, the \texttt{helper} moves on to Step~2, and additionally, the minimum-ID \texttt{helper} takes on the role of \texttt{guard}. If both ports of $v$ are confirmed safe, the \texttt{helper} waits for any active \texttt{checker} to return and then executes $\textsc{Divide}()$. If neither condition holds, the \texttt{helper}'s behavior depends on whether it is assisting a \texttt{settler} or a \texttt{guard}. If assisting a \texttt{settler} and exactly one port has unknown safety status and no assigned \texttt{checker}, the minimum-ID \texttt{helper} at $v$ becomes the \texttt{checker} for that port. If both ports satisfy this condition, the minimum-ID and second-minimum-ID \texttt{helpers} become \texttt{checkers} for port $0$ and port $1$, respectively.

If assisting a \texttt{guard} and all surviving cops have already gathered at $v$ (the \texttt{guard} has set $balance=1$), the \texttt{helper} learns the total number of surviving cops and the waiting time from the \texttt{guard}. If it is among the smaller-ID half of the \texttt{helpers}, it waits accordingly; otherwise, it moves toward the other \texttt{guard}. Either way, it then moves on to Step~2. If assisting a \texttt{guard} and not all cops have gathered yet ($balance$ is not yet set), the \texttt{helper} does two things. First, if some port has unknown safety status and no assigned \texttt{checker}, the minimum-ID \texttt{helper} at $v$ becomes the \texttt{checker} for that port. Second, the \texttt{helper} reacts when a \texttt{checker} arrives at $v$. If this \texttt{checker} is probing on behalf of another \texttt{guard} with a smaller identifier than its own \texttt{guard}, the \texttt{helper} moves toward that \texttt{guard}. If instead the \texttt{checker} is probing on behalf of its own \texttt{guard}, the \texttt{helper} checks what it encountered. If the \texttt{checker} met a \texttt{settler} that has never had a \texttt{checker} assigned, the minimum-ID \texttt{helper} moves to assist that \texttt{settler} in probing its remaining port. If the \texttt{checker} met the other \texttt{guard}, which has the smaller identifier of the two \texttt{guards}, or the \texttt{helper} sees that its own \texttt{guard} has already set $lead=0$, then the \texttt{helper} moves toward the other \texttt{guard} through the safe zone.

\begin{algorithm}[!htb]
\DontPrintSemicolon
\LinesNumbered
\caption{\texttt{Capture-LR} : Step~1 (\texttt{settler})}
\label{alg:step1_settler}
Let $c$ be any cop at vertex $v$ with $c.step=1$ and $c.role=\texttt{settler}$.\;
\If{$c.safe_0=1$ and $c.safe_1=1$}{
    \uIf{$c.checkID_p\ne\perp$, for any port $p\in\{0,1\}$}{
        remain idle until $c.checkID_p=\perp$\;
    }
    $c.role\leftarrow\texttt{helper}$\;
}
\Else(\tcp*[f]{at least one port is not yet confirmed safe}){
    \uIf{a \texttt{checker} $c_1$ with $c_1.portin=p$ and $c_1.\mathrm{ID}\neq c.checkID_p$ is present at $v$, for some $p\in\{0,1\}$}{
        $c.safe_p\leftarrow1$\;
    }
    \For{each port $p\in\{0,1\}$ satisfying $c.checkID_p\neq\perp$}{
        \uIf{a \texttt{checker} $c'$ with $c'.\mathrm{ID}=c.checkID_p$ and $c'.portin=p$ is present at $v$}{
            \uIf{$c'.meet_p=1$}{
                $c.safe_p\leftarrow1$,\: $c.checkID_p\leftarrow\perp$\;
            }
            \uElseIf{$c'.meet_p=0$ and $c.safe_p=1$}{
                $c.checkID_p\leftarrow\perp$\;
            }
            \uElseIf{$c'.meet_p=0$ and $c.safe_p=\perp$}{
                $c.dist_p\leftarrow c.dist_p+1$,\: $c.time_p\leftarrow0$\;
            }
        }
        \ElseIf{no \texttt{checker} $c''$ satisfying $c''.\mathrm{ID}=c.checkID_p$ is present at $v$}{
            \uIf{$c.time_p<2\cdot c.dist_p$}{
                $c.time_p\leftarrow c.time_p+1$\;
            }
            \Else{
                $c.safe_p\leftarrow0$,\: $c.checkID_p\leftarrow\perp$\;
                \uIf{$c.checkID_q\ne\perp$, where $q=1-p$}{
                    remain idle until $c.checkID_q=\perp$\;
                }
                $c.role\leftarrow\texttt{guard}$\;
            }
        }
    }
    \For{each port $p\in\{0,1\}$ in order satisfying $c.safe_p=\perp$ and $c.checkID_p=\perp$}{
        \If{a \texttt{helper} $c_h$ is present at $v$ with $c_h.\mathrm{ID}\neq c.checkID_0$ and $c_h.\mathrm{ID}\neq c.checkID_1$}{
            $\textsc{Assign}(p)$\;
        }
    }
}
\end{algorithm}

\begin{algorithm}[!htb]
\DontPrintSemicolon
\LinesNumbered
\caption{\texttt{Capture-LR} : Step~1 (\texttt{guard})}
\label{alg:step1_guard}
Let $c$ be any cop at vertex $v$ with $c.step=1$, $c.role=\texttt{guard}$, and $c.safe_q=0$ for port $q$. Let $p=1-q$\;   
\uIf{$c.safe_0=0$ and $c.safe_1=0$}{
    $c.N\leftarrow$ total cops at $v$,\: $c.\bar\rho\leftarrow 0$,\: $c.lead\leftarrow 1$,\: $c.uport\leftarrow 0$,\: $c.step\leftarrow 2$\;
}
\ElseIf{$c.balance = 1$}{
    execute $\textsc{Wait}(c.\bar \rho)$\;
    $c.step\leftarrow 2$\;
}
\Else{
    \uIf{a \texttt{checker} $c'$ is present at $v$ with $c'.portin = p$, $c'.task=\texttt{extcheck}$ and  $c'.\mathrm{ID}\ne c.checkID_p$}{
        \uIf{$c'.operatorID>c.\mathrm{ID}$}{
            $c.lead\leftarrow 1$,\: $c.uport\leftarrow q$,\: $c.\bar \rho \leftarrow c'.dist_p$\;
            execute $\textsc{Wait}(2\cdot c.\bar \rho)$\;
            $c.N\leftarrow(\text{cops at }v)+1$,\: $c.balance \leftarrow 1$ \;    
        }
        \Else{
            $c.lead \leftarrow 0$,\: $c.uport\leftarrow q$,\: $c.\bar \rho \leftarrow c'.dist_p$\;
            remain idle until a cop $c''$ is present at $v$ such that $c''.N \ne \perp$\;
            $c.N\leftarrow c''.N$,\: $c.step\leftarrow 2$\;
        }  
    }
    \uIf{$c.checkID_p=\perp$}{
        \uIf{$c.lead = \perp$ and a \texttt{helper} $c_h$ is present at $v$}{
            execute $\textsc{Assign}(p)$\;
        }    
    }
    \Else(\tcp*[f]{a \texttt{checker} has already been assigned to port $p$}){
        \uIf{a \texttt{checker} $c'$ is present at $v$ with $c'.portin = p$ and $c'.\mathrm{ID}=c.checkID_p$}{
            \uIf{$c'.enctype_p=0$}{
                $c.dist_p\leftarrow c.dist_p+1$, \: $c.time_p\leftarrow 0$\;
            }
            \uElseIf{$c'.enctype_p\in\{1,2\}$}{
                $c.time_p\leftarrow 0$\;
            }
            \Else{
                $c.\bar\rho\leftarrow c.dist_p$\;
                \uIf{$c.\mathrm{ID}<c'.encID$}{
                    $c.lead\leftarrow 1$, \: $c.uport\leftarrow q$ \;
                    execute $\textsc{Wait}(2\cdot c.\bar \rho)$\;
                    $c.N\leftarrow(\text{cops at }v)+1$,\: $c.balance \leftarrow 1$ \; 
                }
                \Else{
                    $c.lead\leftarrow 0$,\: $c.uport\leftarrow q$\;
                    remain idle until a cop $c''$ is present at $v$ such that $c''.N \ne \perp$\;
                    $c.N\leftarrow c''.N$,\: $c.step\leftarrow 2$\;
                }
            }
        }
        \Else(\tcp*[f]{the assigned \texttt{checker} has not yet returned}){
            \uIf{$c.time_p < 2\cdot c.dist_p$}{$c.time_p\leftarrow c.time_p+1$\;}
            \Else{
                $c.safe_p\leftarrow 0$\;
            }
        }
    }
}
\end{algorithm}

\smallskip
\noindent\textbf{Algorithm for \texttt{settler}} (see Algorithm~\ref{alg:step1_settler}).
If any port $p$ has unknown safety status with no assigned \texttt{checker}, and at least one unassigned \texttt{helper} is available at $v$, the \texttt{settler} assigns the minimum-ID such \texttt{helper} as the \texttt{checker} for port $p$ via $\textsc{Assign}(p)$. If both ports of the \texttt{settler}'s current vertex $v$ are already confirmed safe, it becomes a \texttt{helper} if it has no assigned \texttt{checker}; otherwise, it waits for its assigned \texttt{checker} to return before becoming a \texttt{helper}. Otherwise, if at least one port is still not marked safe, it works as follows. If the \texttt{settler} sees a \texttt{checker} (not its assigned \texttt{checker}) enter through some port $p$, it marks port $p$ as safe. Otherwise, if its assigned \texttt{checker} returns through port $p$ (within a certain time window, specified later), the \texttt{settler} decides the safety status of port $p$ based on the \texttt{checker}'s report. If the \texttt{checker} encountered another cop during probing, the \texttt{settler} marks port $p$ as safe. If no cop was encountered and the port was already confirmed safe, the \texttt{settler} de-assigns the \texttt{checker} from the role of \texttt{checker} for port $p$. If the safety status of port $p$ is still unknown, the \texttt{settler} increments $c.dist_p$ and resets $c.time_p$ so that the \texttt{checker} probes one step further. If the assigned \texttt{checker} does not return within $2 \cdot c.dist_p$ rounds, the \texttt{settler} marks port $p$ as unsafe. It then waits for the assigned \texttt{checker} of the other port, if any, to return, after which the \texttt{settler} transitions to \texttt{guard}.

\smallskip
\noindent\textbf{Algorithm for \texttt{guard}} (see Algorithm~\ref{alg:step1_guard}).
The \texttt{guard} is stationed at a boundary vertex with one unsafe port $q$ and port $p=(1-q)$ whose safety is unknown. If port $p$ has no assigned \texttt{checker} and \texttt{guard} $c$ has not yet set its $lead$ parameter ($c.lead=\perp$), it assigns the minimum-ID available \texttt{helper} as the \texttt{checker} for port $p$. If both ports are marked unsafe, $c$ moves on to Step~2. If all surviving cops have already gathered at its vertex, the \texttt{guard} waits $\bar\rho$ rounds for the \texttt{helpers} to split before moving on to Step~2. If port $p$ is not yet marked unsafe, the \texttt{guard} works as follows. If a \texttt{checker} arrives at $v$ through port $p$ while executing $\textsc{Extended\_Check}$, but it is not the \texttt{checker} assigned by $c$, then it has come from the other \texttt{guard}'s vertex carrying that \texttt{guard}'s identifier. If this identifier is larger than $c$'s, $c$ sets $c.lead=1$, waits $2\cdot\bar\rho$ rounds for all cops to arrive, counts them, and marks that all cops have gathered. Otherwise, $c$ sets $c.lead=0$, waits until some cop informs it of the total number of surviving cops, and then moves on to Step~2.

If instead $c$'s assigned \texttt{checker} returns through port $p$, it decides what to do based on the \texttt{checker}'s report. If the \texttt{checker} met no cop, or met a \texttt{helper} or another \texttt{checker}, $c$ increments $c.dist_p$ and resets $c.time_p$ so that the \texttt{checker} probes one step further. If the \texttt{checker} met the \texttt{settler}, $c$ only resets $c.time_p$ so that the \texttt{checker} probes the same distance again. If the \texttt{checker} met the other \texttt{guard}, $c$ decides which of the two \texttt{guards} leads by comparing their identifiers, proceeding exactly as described above, where its action depends on whether its identifier is smaller or larger than the other \texttt{guard}'s identifier. If the assigned \texttt{checker} does not return within $2\cdot c.dist_p$ rounds, $c$ marks port $p$ as unsafe.

\medskip
\noindent\textbf{Step 2 (Expanding the safe zone).}
At the completion of Step~1, one \texttt{guard} is positioned at each boundary vertex of the safe zone, while all remaining alive cops are \texttt{helpers} distributed between these two vertices. Every cop knows the values of $\bar{\rho}$ and $N$ computed during Step~1. Let $f(x)=x+\lfloor\log x\rfloor+2$. Each cop computes $\hat{\rho}=\max\{x\in\mathbb{N}:f(x)\le N\}$, and sets $h=\lfloor\log\hat{\rho}\rfloor$. Step~2 consists of $h$ phases, indexed by $i=1,\ldots,h$. Each of the first $h-1$ phases consists of three sub-phases: advancing, checking, and balancing, while the last phase consists only of the advancing sub-phase, after which the robber is surrounded (see Lemma~\ref{lem:checker-availability}).

Every cop $c$ maintains the counters $c.count_1,c.count_2,c.count_3\in\mathbb{N}$, each initialized to $1$ at the beginning of its corresponding sub-phase. At the end of every round of sub-phase $j$, the counter is updated as $c.count_j\leftarrow c.count_j+1$. Thus, $c.count_j$ always equals the current round number within sub-phase $j$ and is reset to $1$ at the beginning of the next phase.

Each \texttt{helper} $c_h$ maintains $c_h.advance\in\{0,1,2\}$, initially $0$. It is set to $1$ when the \texttt{helper} enters the unsafe zone during the advancing sub-phase, and to $2$ when designated as the messenger by the \texttt{guard} with $c.lead=1$, after which it carries the extension information to the other \texttt{guard} through the safe zone. Whenever $c.advance=1$, the parameter $c_h.hops\in\mathbb{N}$ stores the number of hops moved by $c_h$ into the unsafe zone.

At the beginning of each phase $i$, every cop computes the phase parameters $x_i$, $b_i$, and $d_i$. The parameter $x_i$ denotes the number of cops advancing from each boundary vertex during the advancing sub-phase, $b_i$ is a buffer used to correct rounding when $x_{i-1}$ is odd, and $d_i$ denotes the length of the safe zone at the end of phase $i$. These parameters are initialized as $x_0=\hat{\rho}$, $b_0=1$, and $d_0=\bar{\rho}$, and are updated as:
\begin{equation}\label{eq:x_i d_i}
(x_i,\, b_i) =
\begin{cases}
\left(\dfrac{x_{i-1}}{2},\; b_{i-1}\right) & \text{if } x_{i-1} \text{ is even,}\\
\left(\dfrac{x_{i-1}+1}{2},\; 0\right) & \text{if } x_{i-1} \text{ is odd and } b_{i-1} = 1,\\
\left(\dfrac{x_{i-1}-1}{2},\; 1\right) & \text{if } x_{i-1} \text{ is odd and } b_{i-1} = 0,
\end{cases}
\qquad d_i = d_{i-1} + x_i.
\end{equation}

Each \texttt{guard} retains the parameters $c.lead$, $c.uport$, and $c.safe_p$ from Step~1, since they correctly identify the boundary and the unsafe direction. The parameters $c.checkID_p$, $c.dist_p$, and $c.time_p$ are reset to $\perp$, $0$, and $0$, respectively, at the beginning of Step~2, as the checking process starts afresh. 
In addition, the \texttt{guard} with $c.lead=1$ maintains the parameter $c.extend$, which is set during the checking sub-phase to the number of hops by which the safe zone is extended on its side.
Each \texttt{checker} maintains the parameter $c.dist_p\in\mathbb{N}$, storing the distance currently being probed through port $p$.
Step~2 uses the procedures $\textsc{Assign}(p)$ and $\textsc{MoveTill}(p,\mathit{role})$ defined in Step~1. In addition, we define the following procedures.
\begin{itemize}
    \item $\textsc{Advance}(p,i)$: A cop $c$ at vertex $v$ executing this procedure exits through port $p$, moves one hop per round for exactly $i$ rounds, and then remains at the destination vertex.

    \item $\textsc{Advance\_Check}(p,i)$: A cop $c$ at vertex $v$ executing this procedure traverses through port $p$, advancing one hop per round until reaching distance $i$, and then returns to $v$ along the same path. The parameter $c.meet_p$, initially $0$, is set to $1$ if another cop is present at the destination vertex.

    \item $\textsc{Balance}()$: Let $\alpha$ denote the number of \texttt{helpers} at the current vertex. A cop $c$ executing this procedure moves toward the other \texttt{guard} through the safe zone if it is among the $ \lfloor\alpha/2\rfloor $ \texttt{helpers} with the largest identifiers; otherwise, it remains at the current vertex.
\end{itemize}

\smallskip
\noindent\textbf{Algorithm for Sub-phase~1: Advancing} (see Algorithm~\ref{alg:step2a}).
This sub-phase proceeds identically at both boundary vertices. In the first round, the minimum-ID \texttt{helper} departs through the unsafe port. In each subsequent round, the minimum-ID \texttt{helper} currently present at the boundary vertex departs through the same port, while every previously departed \texttt{helper} advances one hop further in the unsafe direction. Thus, over $x_i$ rounds, exactly one \texttt{helper} leaves the boundary vertex in each round, and the movement proceeds in a pipelined manner.
As a special case, both \texttt{guards} may be located at the same vertex, which occurs only in Phase~1 for a rooted initial configuration. In each round, the minimum-ID \texttt{helper} departs through the unsafe port of the \texttt{guard} with $c.lead=1$, while the second-minimum-ID \texttt{helper} departs through the unsafe port of the other \texttt{guard}. At the end of round $x_i+1$, all cops set $c.subphase\leftarrow2$ and begin the next sub-phase in the following round.

\begin{algorithm}[!htb]
\setlength{\algomargin}{0em}
\DontPrintSemicolon
\LinesNumbered
\caption{\texttt{Capture-LR}: Step~2, Phase~$i$ - Sub-phase~1 (Advancing)}
\label{alg:step2a}
Let $c$ be a cop at vertex $v$ with $c.step=2$ and $c.subphase=1$.\;
\uIf{$c.role=\texttt{guard}$}{
    \If{$c.count_1=x_i+1$}{
      $c.subphase\leftarrow 2$
    }
}
\uElseIf{$c.role=\texttt{helper}$}{
    \uIf{$c.count_1=x_i+1$}{
      $c.subphase\leftarrow 2$\;
    }
    \Else{
        \uIf{only one \texttt{guard} $c'$ is present at $v$}{
        \If{$c$ is the minimum-ID helper at $v$}{
            $c.advance\leftarrow 1$\;
            $c.hops\leftarrow x_i - c.count_1+1$\;
            execute $\textsc{Advance}(c'.uport,c.hops)$
        }
    }
        \ElseIf{two \texttt{guards} $c_1$ and $c_2$ are present at $v$}{
        let $c_1$ be the guard with $c_1.lead=1$\;
            \If{$c$ is the minimum-ID helper at $v$}{
                $c.advance\leftarrow 1$\;
                $c.hops\leftarrow x_i - c.count_1+1$\;
                execute $\textsc{Advance}(c_1.uport,c.hops)$\;
            }
            \ElseIf{$c$ is the second minimum-ID helper at $v$}{
                $c.advance\leftarrow 1$\;
                $c.hops\leftarrow x_i - c.count_1+1$\;
                execute $\textsc{Advance}(c_2.uport,c.hops)$\;
            }
        }
    }
}
\end{algorithm}
\begin{algorithm}
\DontPrintSemicolon
\LinesNumbered
\caption{\texttt{Capture-LR}: Step~2, Phase~$i$ - Sub-phase~2 (Checking)}
\label{alg:step2b}
Let $c$ be a cop at vertex $v$ with $c.step=2$ and $c.subphase=2$, and let $T_i=2(1+2+\cdots+x_i)+2x_i+d_i+1$.\;
\BlankLine
\uIf{$c.role=\texttt{guard}$}{
    let $p=c.uport$ and $q=1-p$\;
     \If{$c.count_2=T_i+1$}{
        $c.checkID_p\leftarrow\perp$, \: $c.subphase\leftarrow 3$\;
    }
    \Else{
        \uIf{$c.lead=1$}{
        \uIf{$c.checkID_p=\perp$}{
            $\textsc{Assign}(p)$\;
        }
        \Else{
            \uIf{ a \texttt{checker} $c'$ is present at $v$ with $c'.portin=p$ and $c'.meet_p=1$}{
                \uIf{$c.dist_p<x_i$}{
                    $c.dist_p\leftarrow c.dist_p+1$, \: $c.time_p\leftarrow 0$\;
                }
                \Else{
                    $c.extend\leftarrow x_i$\;
                    execute $\textsc{Advance}(p,x_i)$\;
                    $c.uport\leftarrow 1 -c.portin$\;
                    remain idle until $c.count_2=T_i$\;
                }
            }
            \uElseIf{a \texttt{checker} $c'$ is present at $v$ with $c'.portin=p$ and $c'.meet_p=0$}{
                $c.extend\leftarrow c.dist_p-1$\;
                execute $\textsc{Advance}(p,c.extend)$\;
                $c.uport\leftarrow 1-c.portin$\;
                remain idle until $c.count_2=T_i$\;
}
            \Else{
                $c.time_p\leftarrow c.time_p+1$\;
                \If{$c.time_p=2\cdot c.dist_p$}{
                    $c.extend\leftarrow c.dist_p-1$\;
                    execute $\textsc{Advance}(p,c.extend)$\;
                    $c.uport\leftarrow 1 -c.portin$\;
                    remain idle until $c.count_2=T_i$\;
                }
            }
        }
    }
        \Else{
        \uIf{$c.count_2 \le T_i-x_i$ and a helper $c_h$ with $c_h.advance=2$ is present at $v$}{
            \uIf{$c_h.hops=x_i$}{remain idle until $c.count_2=T_i$\;}
            \Else{
                execute $\textsc{Advance}(c.uport,x_i-c_h.hops)$\;
                $c.uport\leftarrow 1 -c.portin$\;
                remain idle until $c.count_2=T_i$\;
            }
        }
        \ElseIf{$c.count_2=T_i-x_i$ and no helper with $advance=2$ is present at $v$}{
            execute $\textsc{Advance}(c.uport,x_i)$\;
            $c.uport\leftarrow 1 -c.portin$\;
        }
    }
    }
}
\BlankLine
\uElseIf{$c.role=\texttt{helper}$}{
    \If{$c.count_2=T_i+1$}{
        $c.subphase\leftarrow 3$\;
    }
    \Else{
        \uIf{$c.advance=0$}{
        \If{a guard $c_g$ with $c_g.lead=1$ is present at $v$ and $c_g.checkID_p=\perp$ , where $p=c_g.uport$} {
        \uIf{$c$ is minimum-ID \texttt{helper} present at $v$}{$c.role\leftarrow\texttt{checker}$,\: $c.dist_p\leftarrow 1$\;
        execute $\textsc{Advance\_Check}(p,1)$\;}
        }
    }
        \uElseIf{$c.advance=1$}{
        \If{a guard $c_g$ with $c_g.lead=1$ is present at $v$ and $c.hops=c_g.extend$}{
            $c.advance\leftarrow 2$\;
            execute $\textsc{MoveTill}(c.portin,\texttt{guard})$\;
        }
        
    }
    } 
}
\BlankLine
\ElseIf{$c.role=\texttt{checker}$}{
    \If{$c.count_2=T_i+1$}{$c.subphase\leftarrow 3$\;}
    \Else{
        \If{guard $c_g$ is present at $v$ and $c_g.checkID_p=c.\mathrm{ID}$ for some port $p$}{
        \uIf{$c.dist_p\leq x_i$}{
        $c.dist_p \leftarrow c.dist_p+1$\;
        execute $\textsc{Advance\_Check}(p,c.dist_p)$\;}
        \Else{
            $c.role\leftarrow\texttt{helper}$,\: $c.advance\leftarrow 0$\;
            remain idle until $c.count_2=T_i$\;
        }
    }
    }
}
\end{algorithm}

\noindent\textbf{Algorithm for Sub-phase~2: Checking} (see Algorithm~\ref{alg:step2b}).
This sub-phase determines the extension of the safe zone on both sides.
At the vertex of the \texttt{guard} $c_g$ with $c_g.lead=1$, the minimum-ID \texttt{helper} is assigned as the \texttt{checker} for the unsafe port $p=c_g.uport$ in the first round. The \texttt{checker} initially probes one hop through the unsafe port and returns, reporting whether the \texttt{helper} stationed there is still present. As long as the reported vertex remains occupied, the \texttt{checker} probes one hop farther than before. This continues until the \texttt{checker} either finds an empty vertex, fails to return in time, or safely probes the full distance $x_i$. In the first case, $c_g$ determines the extension of the safe zone on its side, stores it in $c_g.extend$, moves to the new boundary vertex, updates its unsafe port, and instructs the \texttt{helper} stationed there to carry $c_g.extend$ through the safe zone to the other \texttt{guard}.

The \texttt{guard} $c'_g$ with $c'_g.lead=0$ waits for this \texttt{helper}. Upon its arrival, $c'_g$ learns $c_g.extend$ and advances $x_i-c_g.extend$ hops through its unsafe port. Since the maximum time for the \texttt{helper} to arrive is $T_i-x_i$, if it does not arrive within this time, $c'_g$ concludes that the safe zone was not extended on the other side and therefore advances the full $x_i$ hops on its own side. Thus, the checking sub-phase completes within $T_i$ rounds, and in round $T_i+1$, every cop sets $c.subphase\leftarrow3$.

\begin{algorithm}[!htb]
\DontPrintSemicolon
\LinesNumbered
\caption{\texttt{Capture-LR}: Step~2, Phase~$i$ - Sub-phase~3 (Balancing)}
\label{alg:step2c}
Let $c$ be a cop at vertex $v$ with $c.step=2$ and $c.subphase=3$.\;
\uIf{$c.role=\texttt{guard}$}{
    \If{$c.count_3=3d_i+1$}{
        $c.phase\leftarrow i+1$\;
        $c.subphase\leftarrow 1$\;
    }
}
\uElseIf{$c.role=\texttt{helper}$}{
    \If{$c.count_3=3d_i+1$}{
        $c.phase\leftarrow i+1$\;
        $c.subphase\leftarrow 1$\;
    }
    \Else{
     \uIf{$c$ is at the vertex of \texttt{guard} $c_g$ with $c_g.lead=1$}{
            remain idle until $c.count_3=2d_i$\;
            execute $\textsc{Balance}()$\;
        }
        \Else{
                execute $\textsc{MoveTill}(c.portin,\texttt{guard})$\;
                \uIf{guard $c_g$ is present at $v$ with $c_g.lead=0$}{
                execute $\textsc{MoveTill}(1-c_g.uport,\texttt{guard})$\;
                }
        }
    }
}
\end{algorithm}

\vspace{-0.3cm}
\noindent\textbf{Algorithm for Sub-phase~3: Balancing} (see Algorithm~\ref{alg:step2c}).
This sub-phase gathers all \texttt{helpers} at the vertex of the \texttt{guard} $c_g$ with $c_g.lead=1$, and then redistributes them between the two new boundary vertices.
Every \texttt{helper} not already at the vertex of $c_g$ moves through the port by which it last entered its current vertex and continues until reaching a \texttt{guard}. If the encountered \texttt{guard} is $c_g$, the \texttt{helper} waits until $c_g$ sets $c_g.balance=1$. Otherwise, it continues through the safe port of the other \texttt{guard} until reaching $c_g$. Within $2d_i$ rounds, all \texttt{helpers} gather at the vertex of $c_g$, after which $c_g$ sets $c_g.balance=1$. In the next round, every \texttt{helper} orders itself by identifier. The \texttt{helpers} with the larger identifiers move through the safe zone toward the other \texttt{guard}, while the remaining \texttt{helpers} wait for $d_i$ rounds, allowing the moving \texttt{helpers} to complete their journey. In the following round, every cop sets $c.phase\leftarrow i+1$ and $c.subphase\leftarrow1$.

\section{Correctness and complexity analysis}
In this section we provide the correctness of \textsc{Capture-LR} along with complexity analysis.
\begin{lemma}\label{lem:elim}
At most two cops are eliminated during Step~1 of the algorithm \texttt{Capture-LR}.
\end{lemma}
\begin{proof}
We first show that only \texttt{checkers} can be eliminated during Step~1. A \texttt{settler} never leaves its current vertex (see Algorithm~\ref{alg:step1_settler}). A \texttt{guard} also remains at its boundary vertex throughout Step~1 and either waits or communicates only with cops arriving at its current vertex (see Algorithm~\ref{alg:step1_guard}). A \texttt{helper} moves only through ports that have already been verified safe, either while executing $\textsc{Divide}()$ (see Algorithm~\ref{alg:step1_helper}, Lines~7--10) or $\textsc{MoveTill}()$ (see Algorithm~\ref{alg:step1_helper}, Lines~22--27 and~37--45). Hence, none of these roles can encounter the robber. The only role that probes unexplored directions is the \texttt{checker} (see Algorithm~\ref{alg:step1_checker}), and therefore only a \texttt{checker} can be eliminated.

A \texttt{checker} can be eliminated only while probing through a port whose safety has not yet been determined, and that leads toward the robber. If the assigned \texttt{checker} does not return within the expected time, the corresponding \texttt{settler} or \texttt{guard} marks that port as unsafe. Moreover, if the cop that assigned the \texttt{checker} is a \texttt{settler}, it becomes a \texttt{guard}, if it has not already done so (see Algorithm~\ref{alg:step1_settler}, Lines~17--24, and Algorithm~\ref{alg:step1_guard}, Lines~36--40). Once a port is marked unsafe, no further \texttt{checker} is assigned to probe through that port (see Algorithm~\ref{alg:step1_guard}). Therefore, each unsafe port can cause at most one \texttt{checker} to be eliminated. Since the ring has exactly two boundary vertices, there are at most two such unsafe ports leading into the unsafe zone. Hence, at most two cops are eliminated during Step~1.
\end{proof}

\begin{lemma}\label{lem:ns_condition}
For Step~1 of \textsc{Capture-LR} to terminate successfully, the initial configuration must contain either a vertex with at least three cops or at least three distinct vertices each containing at least two cops.
\end{lemma}
\begin{proof}
Suppose the initial configuration does not satisfy the stated condition. Then every occupied vertex contains at most two cops, and at most two vertices contain exactly two cops. Consider an initial configuration in which these two vertices are precisely the boundary vertices, and at both vertices, port~$0$ leads toward the unsafe zone. Each such vertex initially contains one \texttt{settler} and one \texttt{helper}. The \texttt{settler} assigns the \texttt{helper} as a \texttt{checker} to probe port~$0$ (see Algorithm~\ref{alg:step1_settler}, Lines~25--27).

Since both \texttt{checkers} move toward the robber, the robber can eliminate both \texttt{checkers}. Consequently, both \texttt{settlers} mark port~$0$ as unsafe and become \texttt{guards} (see Algorithm~\ref{alg:step1_settler}, Lines~17--24). At this point, no \texttt{helper} remains at either boundary vertex to continue probing the remaining port or to propagate information through the safe zone. Every other occupied vertex contains only a single \texttt{settler}, which cannot assign a \texttt{checker}. Hence, no further progress is possible, and Step~1 cannot terminate successfully.

Therefore, for Step~1 of \textsc{Capture-LR} to terminate successfully, the initial configuration must contain either a vertex with at least three cops or at least three distinct vertices each containing at least two cops.
\end{proof}

\begin{lemma}\label{lem:L_condition_step1}
The condition of Lemma~\ref{lem:ns_condition} is always satisfied whenever at least $n-\rho+4$ cops are initially available.
\end{lemma}
\begin{proof}
Recall that $\rho$ denotes the length of the unsafe zone in the initial configuration. Hence, the cops initially occupy at most $n-(\rho-1)=n-\rho+1$ vertices.

Suppose, for contradiction, that the initial configuration does not satisfy the condition of Lemma~\ref{lem:ns_condition}. Then at most two occupied vertices contain two cops each, and every other occupied vertex contains exactly one cop. Since the cops occupy at most $n-\rho+1$ vertices, the total number of cops is at most $2\cdot2+(n-\rho-1)\cdot1=n-\rho+3$, contradicting the assumption that at least $n-\rho+4$ cops are initially available. Therefore, the initial configuration satisfies the condition of Lemma~\ref{lem:ns_condition}.
\end{proof}

\begin{lemma}\label{lem:same_parameters}
At the completion of Step~1, all alive cops possess the same values of the parameters $\bar{\rho}$ and $N$, and simultaneously begin the execution of Step~2.
\end{lemma}
\begin{proof}
The two \texttt{guards} first discover each other in one of two ways. Either a \texttt{guard}'s assigned \texttt{checker} returns after encountering the other \texttt{guard} (\texttt{checker} $c'$ returned through port $p$ with parameter $c'.enctype_p=3$), or a \texttt{checker} assigned by the other \texttt{guard} arrives at its current vertex. In either case, both \texttt{guards} determine the length of the safe zone as the probing distance of the corresponding \texttt{checker} and store it locally as $\bar{\rho}$ (see Algorithm~\ref{alg:step1_guard}, Lines~26--27 and~8--10). Thus, both \texttt{guards} obtain the same value of $\bar{\rho}$.

In either case, the \texttt{guard} with the smaller identifier sets $lead=1$, and the \texttt{guard} with the larger identifier sets $lead=0$ (see Algorithm~\ref{alg:step1_guard}, Lines~8--14 and~26--33). The \texttt{guard} with $lead=1$ then executes $\textsc{Wait}(2\bar{\rho})$, allowing every alive cop, except the other \texttt{guard}, sufficient time to reach its vertex. It then computes the total number of alive cops as $N=(\text{cops at its current vertex})+1$ and sets its parameter  $balance=1$ (see Algorithm~\ref{alg:step1_guard}, Lines~11--12 and~30--31).

Every \texttt{helper} at that vertex observes the \texttt{guard} set the parameter $balance=1$ and  immediately learns the values of $\bar{\rho}$ and $N$ from the \texttt{guard} (see Algorithm~\ref{alg:step1_helper}, Lines~22--23); this happens for all such \texttt{helpers} before any of them move. Only afterward are the \texttt{helpers} divided into two groups by identifier: the $\lceil(N-2)/2\rceil$ \texttt{helpers} with the smaller identifiers remain at the current vertex and execute $\textsc{Wait}(\bar{\rho})$, while the remaining \texttt{helpers}, already carrying the values of $\bar{\rho}$ and $N$, move toward the other \texttt{guard} through the safe zone (see Algorithm~\ref{alg:step1_helper}, Lines~24--27). Since the length of the safe zone is $\bar{\rho}$, this journey takes exactly $\bar{\rho}$ rounds. Correspondingly, the \texttt{guard} with $lead=1$ itself executes $\textsc{Wait}(\bar{\rho})$ once $balance=1$ is set, before proceeding to Step~2 (see Algorithm~\ref{alg:step1_guard}, Lines~4--6). 
Hence, the \texttt{guard} with $lead=1$ and the \texttt{helpers} that remain at its vertex both become ready for Step~2 exactly $\bar{\rho}$ rounds after its parameter $balance$ is set to $1$ — precisely when the departing \texttt{helpers} complete their journey through the safe zone.

Meanwhile, the \texttt{guard} with $lead=0$ remains idle until some cop with a known value of $N$ arrives at its vertex (see Algorithm~\ref{alg:step1_guard}, Lines~15 and~34); this is precisely the moment when the departing \texttt{helpers} arrive, already carrying the correct value of $N$. Upon this arrival, the \texttt{guard} with $lead=0$ sets its own $N$ accordingly and proceeds to Step~2 (see Algorithm~\ref{alg:step1_guard}, Lines~16 and~35).

Therefore, every alive cop possesses the same values of $\bar{\rho}$ and $N$, and all alive cops begin the execution of Step~2 simultaneously.
\end{proof}

\begin{lemma}
\label{lem:rho-determine}
If at least $\rho + \lfloor \log \rho \rfloor + 2$ cops initiate Step~2 of \texttt{Capture-LR}, then every cop can uniquely determine a value $\hat{\rho}$ satisfying $\hat{\rho} \geq \rho$.
\end{lemma}
\begin{proof}
By Lemma~\ref{lem:same_parameters}, every alive cop knows the same value of $N$ at the beginning of Step~2. Let $f(x)=x+\lfloor\log x\rfloor+2$, and consider the set $X=\{x\in\mathbb{N}:f(x)\le N\}$. Since at least $\rho+\lfloor\log\rho\rfloor+2$ cops initiate Step~2, we have $N\ge \rho+\lfloor\log\rho\rfloor+2=f(\rho)$. Hence, $\rho\in X$, implying that the set $X$ is non-empty. Furthermore, $f$ is a strictly increasing function. Also, for any fixed value of $N$, there are only finitely many integers $x$ satisfying $f(x)\le N$. Therefore, the set $X$ has a unique maximum element. Every cop computes $\hat{\rho}=\max X$. Since every cop knows the same value of $N$, they all construct the same set $X$ and consequently compute the same value of $\hat{\rho}$. Finally, because $\rho\in X$ and $\hat{\rho}=\max X$, this implies $\hat{\rho}\ge\rho$.
\end{proof}

\begin{lemma}\label{lem:advancing}
Throughout the execution of Step~2 of \texttt{Capture-LR}, the number of \texttt{helpers} that leave the boundary vertices and move toward the unsafe zone during the advancing sub-phases is at most $\hat{\rho}+1$.
\end{lemma}
\begin{proof}
In the advancing sub-phase of every phase $i\ge1$, $x_i$ \texttt{helpers} from each boundary vertex set $\text{advance}=1$ and leave the boundary vertices, moving one by one toward the unsafe zone in a pipe-lined manner (see Algorithm~\ref{alg:step2a}, Lines~9--23). Thus, a total of $2x_i$ \texttt{helpers} leave the boundary vertices during the advancing sub-phase of phase $i$. We refer to these \texttt{helpers} as the \emph{advancers} of phase $i$, and to their movement as the \emph{advancing task} of phase $i$.

Since the cops and the robber move in alternating rounds, the robber remains stationary during each round in which the cops move. The \emph{advancing task} proceeds simultaneously from both boundary vertices. If the robber is not surrounded, then cops cannot occupy both neighboring vertices of the robber simultaneously. Since a cop can be eliminated only when it moves onto the robber's current vertex, cops advancing from at most one side can be eliminated in any round of the \emph{advancing task}. Hence, among the $2x_i$ \emph{advancers} of phase $i$, at most $x_i$ are eliminated, and at least $x_i$ remain alive at the end of the advancing sub-phase of phase $i$. We refer to these remaining \emph{advancers} as the \emph{surviving advancers} of phase $i$.

Recall the recursive definition of $x_i$ and $b_i$ from Equation~\eqref{eq:x_i d_i}. Initially, $x_0=\hat{\rho}$ and $b_0=1$. Thereafter, for every phase $i\ge1$,
$$
(x_i,b_i)=
\begin{cases}
\left(\dfrac{x_{i-1}}{2},\,b_{i-1}\right), & \text{if }x_{i-1}\text{ is even},\\
\left(\dfrac{x_{i-1}+1}{2},\,0\right), & \text{if }x_{i-1}\text{ is odd and }b_{i-1}=1,\\
\left(\dfrac{x_{i-1}-1}{2},\,1\right), & \text{if }x_{i-1}\text{ is odd and }b_{i-1}=0.
\end{cases}
$$

Let $0\le j_1<j_2<\cdots<j_m\le h$, for some $m\ge0$, denote all indices for which $x_{j_p}$ is odd, where $p\in\{1,2,\ldots,m\}$.

For every $1\le i\le j_1$, we have $x_i=x_{i-1}/2$ and $b_i=1$. Hence, $2x_1=\hat{\rho}$, $2x_2=x_1,\ldots,2x_{j_1}=x_{j_1-1}$. Therefore, for every $2\le i\le j_1$, the \emph{advancing task} of phase $i$ is performed entirely by the \emph{surviving advancers} of phase $i-1$. Consequently, only the initial $\hat{\rho}$ \emph{advancers} have ever left the boundary vertices during the advancing sub-phases up to and including phase $j_1$. Thus, at the end of phase $j_1$, at least $x_{j_1}$ \emph{surviving advancers} of phase $j_1$ remain alive.

For phase $i=(j_1+1)$, since $x_{j_1}$ is odd and $b_{j_1}=1$, we have $x_{j_1+1}=(x_{j_1}+1)/2$ and $b_{j_1+1}=0$. Hence, $2x_{j_1+1}=x_{j_1}+1$. Therefore, the \emph{advancing task} of phase $(j_1+1)$ is performed by the $x_{j_1}$ \emph{surviving advancers} of phase $j_1$ together with one additional \texttt{helper} that has not previously left a boundary vertex during any advancing sub-phase. Thus, at the end of phase $(j_1+1)$, at most $\hat{\rho}+1$ \texttt{helpers} have ever left the boundary vertices during the advancing sub-phases, among which at least $x_{j_1+1}$ remain as the \emph{surviving advancers} of phase $(j_1+1)$.

For every $j_1+2\le i\le j_2$, we have $x_{i-1}$ even and hence $2x_i=x_{i-1}$ and $b_i=b_{i-1}=0$. Therefore, for every such phase, the \emph{surviving advancers} of phase $(i-1)$ exactly suffice to perform the \emph{advancing task} of phase $i$. Hence, at the end of phase $j_2$, at least $x_{j_2}$ \emph{surviving advancers} of phase $j_2$ remain alive among the previously accounted $\hat{\rho}+1$ \texttt{helpers}.

For phase $i=(j_2+1)$, since $x_{j_2}$ is odd and $b_{j_2}=0$, we have $x_{j_2+1}=(x_{j_2}-1)/2$ and $b_{j_2+1}=1$. Hence, $2x_{j_2+1}=x_{j_2}-1$. Therefore, the \emph{advancing task} of phase $(j_2+1)$ is performed by only $x_{j_2}-1$ of the $x_{j_2}$ \emph{surviving advancers} of phase $j_2$. Consequently, one of these \emph{surviving advancers} does not participate in the \emph{advancing task} and remains unused. Hence, at the end of phase $(j_2+1)$, exactly $x_{j_2+1}$ \emph{surviving advancers} of phase $(j_2+1)$, together with the one unused cop, remain among the previously accounted $\hat{\rho}+1$ \texttt{helpers}.

For every $j_2+2\le i\le j_3$, we have $x_{i-1}$ even and hence $2x_i=x_{i-1}$ and $b_i=b_{i-1}=1$. Therefore, for every such phase, the \emph{surviving advancers} of phase $(i-1)$ exactly suffice to perform the \emph{advancing task} of phase $i$. Hence, at the end of phase $j_3$, exactly $x_{j_3}$ \emph{surviving advancers} of phase $j_3$, together with the one unused cop, remain among the previously accounted $\hat{\rho}+1$ \texttt{helpers}.

The above argument extends similarly to the remaining phases. More precisely, for every $j_{2q-1}+2\le i\le j_{2q}$, where $q\ge1$ and $j_{2q}$ exists, we have $x_{i-1}$ even and hence $2x_i=x_{i-1}$ and $b_i=b_{i-1}=0$. Therefore, for every such phase, the \emph{surviving advancers} of phase $(i-1)$ exactly suffice to perform the \emph{advancing task} of phase $i$. Hence, at the end of phase $j_{2q}$, exactly $x_{j_{2q}}$ \emph{surviving advancers} of phase $j_{2q}$ remain alive.

For phase $i=(j_{2q}+1)$, since $x_{j_{2q}}$ is odd and $b_{j_{2q}}=0$, we have $x_{j_{2q}+1}=(x_{j_{2q}}-1)/2$ and $b_{j_{2q}+1}=1$. Hence, $2x_{j_{2q}+1}=x_{j_{2q}}-1$. Therefore, one of the \emph{surviving advancers} of phase $j_{2q}$ does not participate in the \emph{advancing task} of phase $(j_{2q}+1)$ and remains unused.

For every $j_{2q}+2\le i\le j_{2q+1}$, where $j_{2q+1}$ exists, we have $x_{i-1}$ even and hence $2x_i=x_{i-1}$ and $b_i=b_{i-1}=1$. Therefore, for every such phase, $x_i$ of the \emph{surviving advancers} of phase $(i-1)$ suffice to perform the \emph{advancing task} of phase $i$. Hence, at the end of phase $j_{2q+1}$, exactly $x_{j_{2q+1}}$ \emph{surviving advancers} of phase $j_{2q+1}$, together with the previously unused cop, remain among the previously accounted $\hat{\rho}+1$ \texttt{helpers}.

Finally, for phase $i=(j_{2q+1}+1)$, since $x_{j_{2q+1}}$ is odd and $b_{j_{2q+1}}=1$, we have $x_{j_{2q+1}+1}=(x_{j_{2q+1}}+1)/2$ and $b_{j_{2q+1}+1}=0$. Hence, $2x_{j_{2q+1}+1}=x_{j_{2q+1}}+1$. Therefore, the previously unused \emph{advancer} together with $x_{j_{2q+1}}$ of the \emph{surviving advancers} of phase $j_{2q+1}$ provide the $2x_{j_{2q+1}+1}$ \emph{helpers} required for the \emph{advancing task} of phase $(j_{2q+1}+1)$.

Consequently, throughout the execution of Step~2, the total number of \texttt{helpers} that ever leave the boundary vertices and move toward the unsafe zone during the advancing sub-phases is at most $\hat{\rho}+1$.
\end{proof}

\begin{lemma}\label{lem:checker_helper}
If at least $\rho+\lfloor \log \rho \rfloor+2$ cops initiate Step~2, then at the beginning of every checking sub-phase of Step~2, there exists at least one \texttt{helper} at the vertex occupied by the \texttt{guard} $c_g$ with $c_g.\text{lead}=1$.
\end{lemma}
\begin{proof}
Recall that $N$ denotes the number of cops initiating Step~2. By Lemma~\ref{lem:rho-determine}, since at least $\rho+\lfloor\log\rho\rfloor+2$ cops initiate Step~2, every cop determines the same value $\hat{\rho}\ge\rho$ satisfying $N\ge\hat{\rho}+\lfloor\log\hat{\rho}\rfloor+2$.

Among the $N$ cops initiating Step~2, exactly two are \texttt{guards}. Hence, there are $N-2$ \texttt{helpers}. By Lemma~\ref{lem:advancing}, at most $\hat{\rho}+1$ \texttt{helpers} ever leave the boundary vertices and move toward the unsafe zone during the advancing sub-phases of Step~2. Therefore, at least $(N-2)-(\hat{\rho}+1)\ge\lfloor\log\hat{\rho}\rfloor-1$ \texttt{helpers} never leave the boundary vertices during any advancing sub-phase. We refer to these \texttt{helpers} as \emph{non-advancers}.

Step~2 consists of $h=\lfloor\log\hat{\rho}\rfloor$ phases, and only the first $h-1=\lfloor\log\hat{\rho}\rfloor-1$ phases contain a checking sub-phase. During each checking sub-phase, exactly one \texttt{helper} changes its role to \texttt{checker} and executes the procedure \textsc{Advance\_Check}, in which it moves into the unsafe zone (see Algorithm~\ref{alg:step2b}). Therefore, at most one \emph{non-advancer} can be eliminated during each checking sub-phase. Since there are $\lfloor\log\hat{\rho}\rfloor-1$ \emph{non-advancers} and the same number of checking sub-phases, at least one \emph{non-advancer} is alive at the beginning of every checking sub-phase. Therefore, it remains to show that at least one \emph{non-advancer} is present at the vertex occupied by the \texttt{guard} $c_g$ with $c_g.\texttt{lead}=1$ at the beginning of every checking sub-phase.

At the completion of Step~1 and after every balancing sub-phase thereafter, all alive \texttt{helpers} are distributed between the two boundary vertices as evenly as possible. If the number of alive \texttt{helpers} is odd, the boundary vertex occupied by the \texttt{guard} $c_g$ with $c_g.\texttt{lead}=1$ receives one additional \texttt{helper} (see Algorithm~\ref{alg:step2c}). Hence, at least one \emph{non-advancer} is present at the vertex occupied by $c_g$ at the beginning of every checking sub-phase.

Therefore, at the beginning of every checking sub-phase of Step~2, there exists at least one \texttt{helper} at the vertex occupied by the \texttt{guard} $c_g$ with $c_g.\texttt{lead}=1$.

\end{proof}

\begin{lemma}\label{lem:synchronous_execution}
Throughout the execution of Step~2, all alive cops execute every phase and every corresponding sub-phase synchronously.
\end{lemma}
\begin{proof}
By Lemma~\ref{lem:same_parameters}, all alive cops initiate Step~2 simultaneously with identical values of $\bar{\rho}$ and $N$. Therefore, by Lemma~\ref{lem:rho-determine}, every alive cop determines the same value of $\hat{\rho}$. Consequently, every alive cop computes identical values of $h$, $x_i$, $b_i$, and $d_i$ for every phase $i$, since these values depend only on $\hat{\rho}$ (see Equation~\eqref{eq:x_i d_i}).

Consider phase~1. The duration of each sub-phase depends only on the values of $x_1$ and $d_1$, which are identical for all alive cops. In particular, every alive cop executes the advancing, checking, and balancing sub-phases for exactly $x_1+1$, $2(1+2+\cdots+x_1)+2x_1+d_1+2$, and $3d_1+1$ rounds, respectively, irrespective of its role during the execution. Since all alive cops start each sub-phase simultaneously and execute it for the same number of rounds, they complete each sub-phase simultaneously. Therefore, all alive cops complete phase~1 together and start phase~2 in the same round.

Now suppose that, for some $i\geq2$, all alive cops start Phase~$i$ simultaneously. Since every alive cop possesses the same values of $x_i$ and $d_i$, the duration of every sub-phase of phase~$i$ is identical for all alive cops. By the same argument, all alive cops complete each sub-phase of phase~$i$ simultaneously and hence start phase~$i+1$ simultaneously. Therefore, by induction, throughout the execution of Step~2, all alive cops execute every phase and every corresponding sub-phase synchronously.
\end{proof}

\begin{lemma}\label{lem:Sh-bound}
Define $S_k=\sum_{i=1}^{k}x_i$, where $x_i$ is defined by Equation~\eqref{eq:x_i d_i}. Then $S_h\geq\hat{\rho}-1$ for $h=\lfloor \log \hat{\rho} \rfloor$.
\end{lemma}
\begin{proof}

We first derive an explicit expression for $S_h$ using the recurrence in Equation~\eqref{eq:x_i d_i}, and then obtain the required bound. 
Let $0 \le j_1 < j_2 < \cdots < j_m \le h$ be the indices such that $x_{j_\ell}$ is odd for all $\ell \in \{1,2,\dots,m\}$.

Since $j_1$ is the first such index, for all $0\le i \le j_1$, we have $x_i = \frac{\hat{\rho}}{2^i}$ and $b_i=1$. Hence, $S_{j_1} = \sum_{i=1}^{j_1} \frac{\hat{\rho}}{2^i}.$

At index $j_1$, the value $x_{j_1}$ is odd and $b_{j_1}=1$, so
$x_{j_1+1} = \frac{x_{j_1}+1}{2}$ and $b_{j_1+1}=0$.
Expanding,

\begin{align}    
x_{j_1+1} &= \frac{\hat{\rho}}{2^{j_1+1}} + \frac{1}{2}, \\
x_{j_1+2} &= \frac{\hat{\rho}}{2^{j_1+2}} + \frac{1}{2^2}, \\
x_{j_1+3} &= \frac{\hat{\rho}}{2^{j_1+3}} + \frac{1}{2^3}, \\ 
&\vdots
\end{align}

Thus, for all $j_1 < i \le j_2$,
$x_i = \frac{\hat{\rho}}{2^i} + \frac{1}{2^{i-j_1}},$ and $b_i=0$. Hence,
$$
S_{j_2} = \sum_{i=1}^{j_2} \frac{\hat{\rho}}{2^i} + \sum_{i=1}^{j_2-j_1} \frac{1}{2^i}.
$$

At index $j_2$, the value $x_{j_2}$ is odd and $b_{j_2}=0$, so $x_{j_2+1} = \frac{x_{j_2}-1}{2}$ and $b_{j_2+1}=1$.
Expanding,
$$
\begin{aligned}
x_{j_2+1} &= \frac{\hat{\rho}}{2^{j_2+1}} + \frac{1}{2^{j_2+1-j_1}} - \frac{1}{2}, \\
x_{j_2+2} &= \frac{\hat{\rho}}{2^{j_2+2}} + \frac{1}{2^{j_2+2-j_1}} - \frac{1}{2^2}, \\
x_{j_2+3} &= \frac{\hat{\rho}}{2^{j_2+3}} + \frac{1}{2^{j_2+3-j_1}} - \frac{1}{2^3}, \\
&\vdots
\end{aligned}
$$

Thus, for all $j_2 < i \le j_3$, $x_i = \frac{\hat{\rho}}{2^i} + \frac{1}{2^{i-j_1}} - \frac{1}{2^{i-j_2}},$ and $b_i=1$. Hence,
$$
S_{j_3} = \sum_{i=1}^{j_3} \frac{\hat{\rho}}{2^i} + \sum_{i=1}^{j_3-j_1} \frac{1}{2^i} - \sum_{i=1}^{j_3-j_2} \frac{1}{2^i}.
$$

Proceeding similarly, we obtain
\begin{flalign}
&&
S_h = \sum_{i=1}^{h} \frac{\hat{\rho}}{2^i}
+ \left(\sum_{i=1}^{h-j_1} \frac{1}{2^i}
- \sum_{i=1}^{h-j_2} \frac{1}{2^i}
+ \cdots
+ (-1)^{m+1} \sum_{i=1}^{h-j_m} \frac{1}{2^i} \right)
&& \label{eq:Sh-expression}
\end{flalign}

Since $j_1 < j_2 < \cdots < j_m$, we have $h-j_1 > h-j_2 > \cdots > h-j_m$, and hence
$$
\sum_{i=1}^{h-j_1} \frac{1}{2^i} \ge \sum_{i=1}^{h-j_2} \frac{1}{2^i}\ge \cdots \ge \sum_{i=1}^{h-j_m} \frac{1}{2^i}.
$$

Therefore,
\begin{flalign}
&&
\left( \sum_{i=1}^{h-j_1} \frac{1}{2^i}-\sum_{i=1}^{h-j_2} \frac{1}{2^i}+\sum_{i=1}^{h-j_3} \frac{1}{2^i}-\cdots+(-1)^{m+1}\sum_{i=1}^{h-j_m} \frac{1}{2^i} \right)\;\ge\; 0.
&&\label{eq:extra-positive}
\end{flalign}

Hence, from Equations~\eqref{eq:Sh-expression} and~\eqref{eq:extra-positive}, we obtain $S_h \ge \sum_{i=1}^{h} \frac{\hat{\rho}}{2^i}$. Since $h=\lfloor \log \hat{\rho} \rfloor$, we have $2^h \le \hat{\rho} < 2^{h+1}$, and hence $\frac{\hat{\rho}}{2^h}<2$. Therefore, $\sum_{i=1}^{h} \frac{\hat{\rho}}{2^i}=\hat{\rho}\left(1-\frac{1}{2^h}\right)>\hat{\rho}-2$. It follows that $S_h>\hat{\rho}-2$. Since $S_h=\sum_{i=1}^{h}x_i$ and each $x_i$ is an integer, $S_h$ is also an integer. Therefore, $S_h\ge\hat{\rho}-1$, proving the lemma.
\end{proof}

\begin{lemma}\label{lem:checker-availability}
Step~2 of \texttt{Capture-LR} can be executed successfully, and the robber is surrounded, whenever at least $\rho+\lfloor \log \rho \rfloor+2$ cops initiate Step~2. 
\end{lemma}
\begin{proof}
Recall that $N$ denotes the number of cops initiating Step~2. By Lemma~\ref{lem:rho-determine}, since at least $\rho+\lfloor \log \rho \rfloor+2$ cops initiate Step~2, every cop determines the same value $\hat{\rho}\geq\rho$ satisfying $N\geq\hat{\rho}+\lfloor \log \hat{\rho} \rfloor+2=f(\hat{\rho})$.
By Lemma~\ref{lem:advancing}, at most $\hat{\rho}+1$ \texttt{helpers} ever leave the boundary vertices and move toward the unsafe zone during the advancing sub-phases of Step~2. This bound already accounts for any cops eliminated by the robber while advancing. In addition, there are $\lfloor \log \hat{\rho} \rfloor-1$ checking sub-phases, and at most one cop, namely the one acting as the \texttt{checker}, may be eliminated in each checking sub-phase; no cop is eliminated during any balancing sub-phase. Two cops permanently remain \texttt{guards} throughout Step~2. Hence, the total number of cops required throughout Step~2 is
$(\hat{\rho}+1)+(\lfloor \log \hat{\rho} \rfloor-1)+2=\hat{\rho}+\lfloor \log \hat{\rho} \rfloor+2=f(\hat{\rho})$. 
Since $N\geq f(\hat{\rho})$, at least this many cops initiate Step~2.

By Lemma~\ref{lem:checker_helper}, a \texttt{helper} is always available at the vertex of the \texttt{guard} with $c_g.\text{lead}=1$ at the beginning of every checking sub-phase. Hence, every checking sub-phase has a \texttt{helper} available to act as the \texttt{checker}. By Lemma~\ref{lem:synchronous_execution}, all alive cops execute every phase and every corresponding sub-phase synchronously. Therefore, Step~2 never stalls and executes successfully to completion.

It remains to show that the robber is surrounded upon the completion of Step~2. At the completion of Step~1, the safe zone contains $n-\rho+1$ vertices. Suppose, for the sake of contradiction, that the robber is not surrounded upon the completion of Step~2. In the advancing sub-phase of every phase $i\geq 1$, $x_i$ \texttt{helpers} from each boundary vertex move simultaneously toward the unsafe zone (see Algorithm~\ref{alg:step2a}, Lines~9--23). Since the robber can eliminate \texttt{helpers} on at most one side in each round, at least $x_i$ \texttt{helpers} survive the advancing sub-phase. Thus, the size of the safe zone increases by at least $x_i$ during phase $i$, and therefore by at least $S_h$ over the $h$ advancing sub-phases. By Lemma~\ref{lem:Sh-bound}, $S_h\geq\hat{\rho}-1$. Since $\hat{\rho}\geq\rho$, we have $S_h\geq\rho-1$. Thus, after the $h$ advancing sub-phases, the safe zone contains at least $n-\rho+1+(\rho-1)=n$ vertices. Therefore, the safe zone covers the entire ring, which implies that the robber is surrounded, a contradiction. Hence, the robber is surrounded upon the completion of Step~2.
\end{proof}

\begin{lemma}\label{lem:complexity}
Algorithm \texttt{Capture-LR} requires $O(\log n)$ bits of memory per cop and surrounds the robber within $O(n^3)$ rounds.
\end{lemma}
\begin{proof}
During the execution of \texttt{Capture-LR}, cops with different roles maintain parameters corresponding to ports, counters, distances, phase variables, and synchronization variables. All such parameters, except those storing cop IDs, take values bounded by a polynomial in $n$. Hence, each such parameter requires at most $O(\log n)$ bits of memory. Certain parameters, such as $checkID$, $encID$, and $operatorID$, store cop IDs. Since each cop ID belongs to the range $[1,n^\lambda]$ for some constant $\lambda$, storing such an ID also requires $O(\log n)$ bits of memory. Moreover, throughout the execution of the algorithm, each cop maintains only a constant number of parameters. Therefore, the total memory required per cop is $O(\log n)$ bits.

We now analyze the time complexity of the algorithm. We first analyze the number of rounds required for the completion of Step~1. During Step~1, a \texttt{checker} executes either the \textsc{Check} or the \textsc{Extended\_Check} procedure (see Algorithms~\ref{alg:step1_checker} and~\ref{alg:step1_helper}). Since cops have bounded speed and can move at most one hop in each round, a \texttt{checker} exploring up to distance $i$ requires $O(i)$ rounds. For a \texttt{settler} located at a non-boundary vertex of the safe zone, the cumulative number of rounds spent in all checking procedures before both corresponding ports are marked safe is $O((n-\rho)^2)$, since the length of the safe zone is $n-\rho$. Once both ports are marked safe, all the cops present at that vertex, together with the assigned \texttt{checkers}, become \texttt{helpers} and move towards the encountered \texttt{settlers}, which requires at most $O(n-\rho)$ further rounds. Upon reaching another \texttt{settler}, a \texttt{helper} may again be assigned as a \texttt{checker} and repeat the checking process from distance $1$. Since at most $n-\rho-1$ non-boundary vertices of the safe zone can give rise to such a checking-and-redistribution process, the total number of rounds contributed across all of them is $O((n-\rho)^3)$.

At the two boundary vertices of the safe zone, the corresponding \texttt{settlers} eventually become \texttt{guards}. Their checking procedures may explore both the safe and unsafe zones. Since the total length of the ring is $n$, the cumulative number of rounds spent in all such checking procedures is bounded by $O(n^2)$. Moreover, the corresponding redistribution procedures require at most $O(n-\rho)$ rounds. Hence, the total number of rounds contributed by the two boundary vertices is $O(n^2+n-\rho)=O(n^2)$. Therefore, the total number of rounds required for the completion of Step~1 is $O((n-\rho)^3)+O(n^2)=O(n^3)$.

We now analyze Step~2. By Lemma~\ref{lem:rho-determine}, every cop computes a value $\hat{\rho}\geq\rho$ and $h=\lfloor\log\hat{\rho}\rfloor$ at the start of Step~2. For each phase $i\in\{1,\dots,h\}$, the advancing sub-phase requires exactly $x_i+1$ rounds, the checking sub-phase requires $2(1+2+\cdots+x_i)+2x_i+d_i+2$ rounds, and the balancing sub-phase requires exactly $3d_i+1$ rounds. Hence, the total number of rounds required in phase $i$ is $(x_i+1)+\bigl(2(1+2+\cdots+x_i)+2x_i+d_i+2\bigr)+(3d_i+1)$. Since $2(1+2+\cdots+x_i)=x_i(x_i+1)$ for every $1\leq i\leq h$, the number of rounds required to complete phase $i$ is $O(x_i^2+d_i)$.

Observe that $\hat{\rho}$ is computed solely from $N$, the number of cops that initiate Step~2.
Since $N$ is independent of the ring size $n$, $\hat{\rho}$ may be either smaller or larger than $2n$. Therefore, we distinguish the following two cases.

\noindent\underline{\textbf{Case 1 ($\hat{\rho}\geq2n$)}}.
By Equation~\eqref{eq:x_i d_i}, $x_1\geq\hat{\rho}/2\geq n$. Thus, during the advancing sub-phase of the first phase, more than $n$ cops advance consecutively from each boundary vertex of the safe zone.
Consequently, irrespective of the initial configuration, the advancing cops traverse the entire unsafe zone, and hence occupy both neighbors of the robber before the first advancing sub-phase terminates. Hence, the robber is surrounded within the first $n$ rounds of Step~2. Therefore, the cops surround the robber within $O(n^3)+O(n)=O(n^3)$ rounds.

\noindent\underline{\textbf{Case 2 ($\hat{\rho}<2n$)}}.
Since $\hat{\rho}=O(n)$, by Equation~\eqref{eq:x_i d_i}, $x_i\leq(x_{i-1}+1)/2$ and $d_i=d_{i-1}+x_i$, where $x_0=\hat{\rho}$ and $d_0=\bar{\rho}$. Hence, $x_i=O(n)$ and $d_i=O(n)$ for every phase $i \in \{1,2,\dots , h\}$. Therefore, the number of rounds required to complete each phase $i$ is $O(x_i^2+d_i)=O(n^2+n)=O(n^2)$. Moreover, since $h=\lfloor\log\hat{\rho}\rfloor$, we have $h=O(\log n)$. By Lemma~\ref{lem:checker-availability}, Step~2 executes successfully and surrounds the robber. Since Step~2 consists of $h$ phases, the robber is surrounded within at most $h$ phases. Hence, Step~2 completes within $O(n^2\log n)$ rounds. Therefore, the cops surround the robber within $O(n^3)+O(n^2\log n)=O(n^3)$ rounds.

Hence, in both cases, the cops surround the robber within $O(n^3)$ rounds.
\end{proof}

\begin{theorem}\label{th:main}
Let $n\geq4$, and let $\rho$ denote the length of the unsafe zone in the initial configuration. Starting from any arbitrary initial configuration on the ring $C_n$, if at least $\max\{\rho+\lfloor \log \rho \rfloor+4,\; n-\rho+4\}$ cops, each equipped with $O(\log n)$ bits of memory, execute the algorithm \texttt{Capture-LR}, then the robber is surrounded within $O(n^3)$ rounds.
\end{theorem}
\begin{proof}
Let $k$ denote the number of cops initially available. By Lemmas~\ref{lem:ns_condition} and~\ref{lem:L_condition_step1}, Step~1 of \texttt{Capture-LR} executes successfully whenever $k\geq n-\rho+4$. Further, by Lemma~\ref{lem:checker-availability}, Step~2 of \texttt{Capture-LR} executes successfully and the robber is surrounded whenever at least $\rho+\lfloor\log\rho\rfloor+2$ cops initiate Step~2. Since, by Lemma~\ref{lem:elim}, at most two cops are eliminated during Step~1, at least $k-2$ cops remain alive at the beginning of Step~2. Hence, Step~2 executes successfully and the robber is surrounded whenever
$k-2\geq\rho+\lfloor\log\rho\rfloor+2$,
that is, whenever
$k\geq\rho+\lfloor\log\rho\rfloor+4$.
Therefore, \texttt{Capture-LR} successfully surrounds the robber whenever
$k\geq\max\{\rho+\lfloor\log\rho\rfloor+4,\;n-\rho+4\}$.
Finally, by Lemma~\ref{lem:complexity},
\texttt{Capture-LR} requires $O(\log n)$ bits of memory per cop, and the cops surround the robber within $O(n^3)$ rounds.
\end{proof}

\begin{remark}
The rooted initial configuration is the worst-case arbitrary initial configuration for \texttt{Capture-LR}, requiring $n+\lfloor\log n\rfloor+4$ cops to successfully surround the robber.
\end{remark}

\section{Discussion and conclusion}\label{sec:disc}
In the classical Black Hole Search (BHS) problem, the black hole is modeled as a static dangerous vertex: any mobile entity entering that vertex is immediately destroyed. A natural extension of this model is to consider a \emph{lethal mobile entity} that can move through the graph over time. In this setting, if a movable resource and the lethal entity attempt to move to the same vertex during a round, then the resource is destroyed. Hence, compared to the classical black hole, the mobile lethal entity possesses additional power due to its mobility. If such an entity is allowed to move freely forever, then mobile entities may continue to be destroyed indefinitely. Therefore, rather than only detecting the dangerous entity, a more meaningful objective is to \emph{capture} or \emph{block} its movement. Observe that the lethal mobile entity can be interpreted as a robber with destructive capability. Suppose we assign all the powers of the lethal entity to the robber in the $\mathcal{CLR}$ model. In our setting, each round consists of two phases: first, the robber moves, and then the cops move. If these two phases are viewed together as a single synchronized round, then the robber behaves exactly as a mobile lethal entity. Indeed, whenever one or more cops attempt to move to the same vertex as the robber during that round, the cops are eliminated.
Consequently, our $\mathcal{CLR}$ algorithm can also be interpreted as an algorithm for capturing a lethal mobile entity. The correctness of the strategy follows from the fact that the cops progressively restrict the movement of the robber until capture becomes inevitable. Hence, our algorithm provides a method for safely capturing a movable black-hole-like entity using mobile resources. 

The immediate open directions include trying for an algorithm with less number of cops and/or providing a better impossibility result. 


\bibliographystyle{plainurl}
\bibliography{Updated_References}
\end{document}